\documentclass[11pt,letterpaper]{article}
\usepackage{geometry}
\def\PrintMode{0}

\usepackage{tikz}
\usetikzlibrary{positioning}
\usepackage{pgffor}
\usetikzlibrary{decorations.pathreplacing}
\usepackage{expl3}

\usepackage[T1]{fontenc}
\usepackage{graphicx}
\usepackage{float}
\usepackage[lined,ruled,linesnumbered]{algorithm2e}
\usepackage{xspace}
\usepackage[utf8]{inputenc}
\usepackage[noadjust]{cite}
\usepackage{amsmath, amssymb, amsfonts, amsthm}
\usepackage{enumerate}
\usepackage{xcolor}
\usepackage[linktocpage=true,pagebackref=true,colorlinks,linkcolor=magenta,citecolor=blue,bookmarks,bookmarksopen,bookmarksnumbered]{hyperref}
\usepackage{bm}
\usepackage[normalem]{ulem}
\usepackage{paralist}
\usepackage{gensymb}
\usepackage{thmtools}
\usepackage{rotating}
\usepackage{mdframed}
\usepackage{multirow}
\usepackage{tabularx}
\usepackage{verbatim}
\usepackage{mathrsfs}
\usepackage{indentfirst}
\usepackage{mleftright}
\usepackage[nottoc]{tocbibind}
\usepackage{complexity}
\usepackage{tablefootnote}
\usepackage{epigraph}
\usepackage{autobreak}

\ifnum\PrintMode=1

\usepackage{savetrees}
\else
\fi

\graphicspath{{./figs/}}

\usepackage[capitalize]{cleveref}

\newtheorem{theorem}{Theorem}[section]
\newtheorem{lemma}[theorem]{Lemma}
\newtheorem{corollary}[theorem]{Corollary}
\newtheorem{claim}[theorem]{Claim}
\newtheorem{proposition}[theorem]{Proposition}

\newtheorem{fact}[theorem]{Fact}

\theoremstyle{definition}
\newtheorem{definition}[theorem]{Definition}

\theoremstyle{remark}
\newtheorem{remark}[theorem]{Remark}

\renewcommand*\backref[1]{\ifx#1\relax \else (cit.~on p.~#1) \fi} 

\makeatletter
\def\moverlay{\mathpalette\mov@rlay}
\def\mov@rlay#1#2{\leavevmode\vtop{%
		\baselineskip\z@skip \lineskiplimit-\maxdimen
		\ialign{\hfil$\m@th#1##$\hfil\cr#2\crcr}}}
\newcommand{\charfusion}[3][\mathord]{
	#1{\ifx#1\mathop\vphantom{#2}\fi
		\mathpalette\mov@rlay{#2\cr#3}
	}
	\ifx#1\mathop\expandafter\displaylimits\fi}
\makeatother

\renewcommand{\emptyset}{\varnothing}

\newcommand{\cd}{\mathsf{cd}}
\newcommand{\Ccd}{\mathsf{Ccd}}

\newlang{\MCSP}{MCSP}
\newlang{\MFSP}{MFSP}
\newlang{\MKtP}{MKtP}
\newlang{\MKTP}{MKTP}
\newlang{\itrMCSP}{itrMCSP}
\newlang{\itrMKTP}{itrMKTP}
\newlang{\itrMINKT}{itrMINKT}
\newlang{\MINKT}{MINKT}
\newlang{\MINK}{MINK}
\newlang{\MINcKT}{MINcKT}
\newlang{\CMD}{CMD}
\newlang{\DCMD}{DCMD}
\newlang{\CGL}{CGL}
\newlang{\PARITY}{PARITY}
\newlang{\Empty}{\textsc{Empty}}
\newlang{\Avoid}{\textsc{Avoid}}
\newlang{\Sparsification}{\textsc{Sparsification}}
\newlang{\HamEst}{\mathsf{HammingEst}}
\newlang{\HamHit}{\mathsf{HammingHit}}
\newlang{\CktEval}{\textsc{Circuit-Eval}}
\newlang{\Hard}{\textsc{Hard}}
\newlang{\cHard}{\textsc{cHard}}
\newlang{\CAPP}{CAPP}
\newlang{\GapUNSAT}{GapUNSAT}
\newlang{\OV}{OV}
\renewlang{\PCP}{PCP}
\newlang{\PCPP}{PCPP}
\newclass{\Avg}{Avg}
\newclass{\ZPEXP}{ZPEXP}
\newclass{\DLOGTIME}{DLOGTIME}
\newclass{\ALOGTIME}{ALOGTIME}
\newclass{\ATIME}{ATIME}
\newclass{\SZKA}{SZKA}
\newclass{\Laconic}{Laconic\text{-}}
\newclass{\APEPP}{APEPP}
\newclass{\SAPEPP}{SAPEPP}
\newclass{\TFSigma}{TF\Sigma}
\newclass{\NTIMEGUESS}{NTIMEGUESS}

\newlang{\Formula}{Formula}
\newlang{\THR}{THR}
\newlang{\EMAJ}{EMAJ}
\newlang{\MAJ}{MAJ}
\newlang{\SYM}{SYM}
\newlang{\DOR}{DOR}
\newlang{\ETHR}{ETHR}
\newlang{\Midbit}{Midbit}
\newlang{\LCS}{LCS}
\newlang{\TAUT}{TAUT}
\newlang{\Poly}{\text{-}Poly}

\newcommand{\maex}{\mathsf{Ex}}

\newcommand{\dist}{\mathsf{dist}}

\renewcommand{\epsilon}{\varepsilon}

\usepackage[draft, inline,marginclue,index]{fixme}  
\fxsetup{theme=color,mode=multiuser}

\definecolor{color1}{RGB}{46,134,193}
\FXRegisterAuthor{hanlin}{ahanlin}{\colorbox{color1}{\color{white}Hanlin}}

\definecolor{color7}{RGB}{128,0,128}
\FXRegisterAuthor{yeyuan}{ayeyuan}{\colorbox{color7}{\color{white}Yeyuan}}

\definecolor{color3}{RGB}{255,128,0}
\FXRegisterAuthor{jiatu}{ajiatu}{\colorbox{color3}{\color{white}Jiatu}}

\definecolor{color4}{RGB}{150,150,150}
\FXRegisterAuthor{tianqi}{atianqi}{\colorbox{color4}{\color{white}Tianqi}}

\definecolor{color2}{RGB}{20,60,100}
\FXRegisterAuthor{zhikun}{azhikun}{\colorbox{color2}{\color{white}Zhikun}}

\definecolor{color6}{RGB}{250,0,250}
\FXRegisterAuthor{yizhi}{ayizhi}{\colorbox{color6}{\color{white}Yizhi}}

\definecolor{color5}{RGB}{128,128,128}
\FXRegisterAuthor{task}{atask}{\colorbox{color5}{\color{white}Task}}

\begin{document}

\title{The Gap Between Greedy Algorithm and Minimum Multiplicative Spanner}
\author{
\text{Yeyuan Chen}\vspace{6pt}\\{EECS Department}\\University of Michigan, Ann Arbor\\ \href{mailto:yeyuanch@umich.edu}{\texttt{yeyuanch@umich.edu}}
}
\maketitle
\begin{abstract}
Given any undirected graph $G=(V,E)$ with $n$ vertices, if its subgraph $H\subseteq G$ satisfies $\dist_{H}(u,v)\leq k\cdot\dist_G(u,v)$ for any two vertices $(u,v)\in V\times V$, we call $H$ a $k$-spanner of $G$. The greedy algorithm adapted from Kruskal's algorithm is an efficient and folklore way to produce a $k$-spanner with girth at least $k+2$. The greedy algorithm has shown to be `existentially optimal', while it's not `universally optimal' for any constant $k$. Here, `universal optimality' means an algorithm can produce the smallest $k$-spanner $H$ given any $n$-vertex input graph $G$.
 
However, how well the greedy algorithm works compared to `universal optimality' is still unclear for superconstant $k:=k(n)$. In this paper, we aim to give a new and fine-grained analysis of this problem in undirected unweighted graph setting. Specifically, we show some bounds on this problem including the following two 
\begin{itemize}
\item On the negative side, when $k<\frac{1}{3}n-O(1)$, the greedy algorithm is not `universally optimal'.
\item On the positive side, when $k>\frac{2}{3}n+O(1)$, the greedy algorithm is `universally optimal'.
\end{itemize}
We also introduce an appropriate notion for `approximately universal optimality'. An algorithm is $(\alpha,\beta)$-universally optimal iff given any $n$-vertex input graph $G$, it can produce a $k$-spanner $H$ of $G$ with size $|H|\leq n+\alpha(|H^*|-n)+\beta$, where $H^*$ is the smallest $k$-spanner of $G$. We show the following positive bounds.
\begin{itemize}
\item When $k>\frac{4}{7}n+O(1)$, the greedy algorithm is $(2,O(1))$-universally optimal.
\item When $k>\frac{12}{23}n+O(1)$, the greedy algorithm is $(18,O(1))$-universally optimal.
\item When $k>\frac{1}{2}n+O(1)$, the greedy algorithm is $(32,O(1))$-universally optimal.
\end{itemize}
All our proofs are constructive building on new structural analysis on spanners. We give some ideas about how to break small cycles in a spanner to increase the girth. These ideas may help us to understand the relation between girth and spanners. 
\end{abstract}
\setcounter{tocdepth}{2}
\pagenumbering{roman}

\pagenumbering{arabic}
\section{Introduction}
A spanner of a graph is a subgraph that approximately preserves distances between all vertex pairs in the original graph. This notion was introduced by \cite{spanner_intro} and widely used in many areas such as graph algorithms \cite{graphalg1,graphalg2}, distributed networks \cite{dis1,dis2} and chip design \cite{chip1,chip2}.  In most applications, people aim to construct spanners as sparse as possible. The sparsity is measured by the number of edges in the spanner.

In this paper, we will focus on multiplicative spanners for undirected unweighted graphs defined below. Unless otherwise stated, all graphs in this paper are undirected unweighted graphs.
\begin{definition}
Given $k:=k(n)$ and an undirected unweighted graph $G:=(V,E)$ with $n$ vertices, a $k$-spanner of $G$ is a subgraph $H=(V,E'),E'\subseteq E$ of $G$ such that for any vertex pair $(u,v)\in V\times V$, there is $\dist_{H}(u,v)\leq k\cdot\dist_{G}(u,v)$. We call $k$ the stretch of $H$.
\end{definition}
One of the simplest and most widely-used spanner construction algorithms is the greedy algorithm introduced by \cite{add93}. It is adapted from Kruskal's algorithm for computing the minimum spanning tree. In undirected unweighted settings, the greedy algorithm runs like this: Given any input graph $G=(V,E)$ and stretch $k$, we first choose an arbitrary total order $\sigma$ on its edges, denoted by $\sigma(E)=(e_1<e_2<\dots<e_m)$. Then, starting from empty subgraph $H=(V,\emptyset)$, we proceed edges $\{e_i\}_{i\in[m]}$ one by one. For each edge $e_i=(u,v)$, if $\dist_{H}(u,v)>k$, we add $e_i$ to edgeset of $H$. Otherwise, we do nothing. After checking all $m$ edges, we output the latest $H$ as the spanner. \cite{add93} proved the following fact about $H$
\begin{proposition}\label{greedyalg}
The subgraph $H$ outputted by the greedy algorithm with input $\langle G,k\rangle$ is a $k$-spanner of $G$. Moreover, the girth of $H$ is at least $k+2$.
\end{proposition}
\Cref{greedyalg} shows that for any graph $G=(V,E)$ with $n$ vertices and $k$, we can use the greedy algorithm to compute one of its $k$-spanners with girth at least $k+2$. A simple observation is that the backward direction also holds. Concretely, we have the following
\begin{proposition}\label{girthequivgreedy}
Given any graph $G=(V,E)$ with $n$ vertices, stretch $k$ and any of its $k$-spanner $H$ with girth at least $k+2$, there exists an edge ordering $\sigma$ of $E$ such that the greedy algorithm outputs $H$ after running on $\sigma$, with input $\langle G,k\rangle$.
\end{proposition}
\begin{proof}
Let $H=(V,E'\subseteq E)$. Consider the edge ordering $\sigma$ on $E$ such that all edges in $E'$ appear earlier in $\sigma$ than any edge in $E\backslash E'$. 
Let $H'$ be the spanner outputted by the greedy algorithm running on $\sigma$. Since we first proceed with edges in $E'$, all edges in $E'$ will be in $H'$ by the girth assumption. For any edge $e=(u,v)\in E\backslash E'$, since $H$ is a $k$-spanner of $G$, there must be $\dist_H(u,v)\leq k$, and therefore $e$ won't be added in $H'$ after all edges in $E'$ have been added. We conclude $H'=H$.
\end{proof}
\Cref{girthequivgreedy} shows that: Given $\langle G=(V,E),k\rangle$, the set of subgraphs the greedy algorithm can output (by choosing arbitrary edge ordering) is exactly the set of $k$-spanners with girth at least $k+2$. By the Moore bound \cite{moore}, all graphs $H$ with girth at least $k+2$ and $n$ vertices have at most $m=O(n^{1+2/k})$ edges. Therefore, we can get a size upper bound of $k$-spanners possibly outputted by the greedy algorithm. There are also papers studying this upper bound on weighted graph setting \cite{add93,light2,light3,light4}. Then, a natural question arises: How large is the gap between the minimum (optimal) spanners and spanners outputted by the greedy algorithm? Here, the minimum $k$-spanner of a graph $G$ is its $k$-spanner with the fewest edges. In a seminal work \cite{existopt}, it is shown that the greedy algorithm is \emph{existentially optimal}. Namely, for any $n,k$, they show that there exists an $n$-vertex graph $G$ and its minimum $k$-spanner $H^*$ such that any $k$-spanner $H$ outputted by the greedy algorithm running on any $n$-vertex input graph is no larger than $H^*$. A more classical and difficult notion is \emph{universally optimal}. It requires the greedy algorithm to output a minimum $k$-spanner for any $n$-vertex graph. However, \cite{existopt} also shows that for many parameter settings $(n,k)$, the greedy algorithm for weighted graphs is `far from'  \emph{universally optimal} on stretch $k$ and $n$-vertex graphs. It means that the greedy algorithm cannot be \emph{universally optimal} on general parameters.

On the other hand, there is still hope that the greedy algorithm is \emph{universally optimal} on some specific 
parameters $(n,k)$. In fact, a simple observation is that when we set $k=n-1$, the $k$-spanners of $n$-vertex graphs are actually equivalent to their reachability preservers, and we can assert that the greedy algorithm is \emph{universally optimal} on parameters $(n,k=n-1)$ since reachability preservers are just spanning trees, which can be outputted by the greedy algorithm. It is natural to ask: for what parameters $(n,k)$, the greedy algorithm is \emph{universally optimal}? That is, it outputs a minimum $k$-spanner for any $n$-vertex graphs. In this paper, we will try to partially answer this question.
\subsection{Our Results}

We first define some terms, which is closely related to the performance of the greedy algorithm.
\begin{definition}
For any pair $(n,k)\in\mathbb{N}\times \mathbb{N}$, we define:
\begin{compactitem}
\item[(a)] If for any $n$-vertex graph, all of its minimum $k$-spanners have girth at least $k+2$, we call $(n,k)$ an \emph{extremely good} pair.
\item[(b)] If for any $n$-vertex graph, at least one of its minimum $k$-spanners has girth at least $k+2$, we call $(n,k)$ a \emph{good} pair.
\item[(c)] If for any $n$-vertex graph $G$, all of its minimum $k$-spanners can be outputted by the greedy algorithm on some edge ordering $\sigma$, we say the greedy algorithm is \emph{completely universally optimal} on $(n,k)$.
\item[(d)] If for any $n$-vertex graph $G$, there exists one of its minimum $k$-spanners that can be outputted by the greedy algorithm on some edge ordering $\sigma$, we say the greedy algorithm is \emph{universally optimal} on $(n,k)$.
\end{compactitem}
\end{definition}
From \Cref{girthequivgreedy}, we can connect the above notions
\begin{corollary}\label{gooddef}For any pair $(n,k)\in\mathbb{N}\times \mathbb{N}$, we have:
\item[(a)] If $(n,k)$ is \emph{extremely good}, then the greedy algorithm is \emph{completely universal optimal} on $(n,k)$.
\item[(b)] If $(n,k)$ is \emph{good}, then the greedy algorithm is \emph{universally optimal} on $(n,k)$.
\end{corollary}
By \Cref{gooddef}, if we want to argue \emph{(completely) universal optimality} on specific $(n,k)$, it suffices to identify the category of $(n,k)$: extremely good, good or `not even good'.
\paragraph{Lower Bound}
First, we give a lower (negative) bound. It asserts that if $k$ is too small compared to $n$, then $(n,k)$ is not even a good pair.  
\begin{theorem}\label{goodlb}
For every sufficiently large $n$ and $O(1)<k<\frac{1}{3}n-O(1)$, $(n,k)$ is not a good pair.
\end{theorem}
In fact, the proof of \Cref{goodlb} is constructive: We can construct a graph $G$ with $n$ vertices such that all of its minimum $k$-spanners have girth at most $k+1$.  
\paragraph{Upper Bound}

Extremely good $(n,k)$ is the strongest definition, so we first find an upper (positive) bound about it. Our proof is algorithmic. Concretely, if $k$ is large enough compared to $n$, for any $k$-spanner $H$ of $n$-vertex graph $G$ with girth at most $k+1$, we can algorithmically remove some edges of $H$ such that $H$ is still a $k$-spanner after removal.
\begin{theorem}\label{exgoodub}
There is a deterministic polynomial time algorithm $A$ such that for all sufficiently large $n$ and any $k>\frac{3}{4}n+O(1)$, given any $n$-vertex graph $G$ and its $k$-spanner $H\subseteq G$, $A(G,H,k)$ outputs a subgraph $R\subseteq H$ of $H$ such that $R$ is a $k$-spanner of $G$ and $R$ has girth at least $k+2$.  
\end{theorem}
This immediately gives us an upper bound for extremely good $(n,k)$.
\begin{corollary}\label{goodexdef}
    For all sufficiently large $n$ and any $k>\frac{3}{4}n+O(1)$, $(n,k)$ is extremely good.
\end{corollary}
\begin{proof}
Suppose by contradiction that for some $n$-vertex $G$, it has some minimum $k$-spanner $H$ with girth at most $k+1$. Then by \Cref{exgoodub} we can construct its subgraph $R\subseteq H$ such that $R$ is also a $k$-spanner of $G$. Since $R$ has strictly larger girth than $H$, $R$ must be a strictly smaller $k$-spanner of $G$ than $H$, which contradicts the assumption that $H$ is the minimum $k$-spanner.
\end{proof}
We can also give an upper bound using a similar but more complicated algorithm for the weaker notion, good pair $(n,k)$.
\begin{theorem}\label{goodub}
    There is a deterministic polynomial time algorithm $A$ such that for all sufficiently large $n$ and any $k>\frac{2}{3}n+O(1)$, given any $n$-vertex graph $G$ and its $k$-spanner $H=(V,E_H)\subseteq G$, $A(G,H,k)$ outputs $R=(V,E_R)\subseteq G$ such that $R$ is a $k$-spanner of $G$ with girth at least $k+2$. Moreover, the size of 
 $R$ satisfies $|E_R|\leq|E_H|$.  
\end{theorem}
Similarly, this immediately gives us an upper bound for good pairs
\begin{corollary}\label{goodupdef}
    For all sufficiently large $n$ and any $k>\frac{2}{3}n+O(1)$, $(n,k)$ is good.
\end{corollary}
\paragraph{Approximately Universally Optimal} By the previous results, we have known whether the greedy algorithm is universally optimal when $k<\frac{1}{3}n-O(1)$ or $k>\frac{2}{3}n+O(1)$. It remains open when $\frac{1}{3}n<k<\frac{2}{3}n$. To understand the gap between the greedy algorithm and minimum spanners in this range, we need to introduce a notion called `approximately universally optimal'.

Since we actually have a polynomial time algorithm in \Cref{goodub}, we can consider slacking the size requirement of outputted spanner $R$. Concretely, rather than requiring $|E_R|\leq |E_H|$, we can output a $k$-spanner $R$ with girth at least $k+2$ that is just `slightly' larger than $H$. In the classical approximation algorithm setting, we use a multiplicative slack factor. Formally, Suppose $H$ is the minimum $k$-spanner, if our outputted large-girth $k$-spanner $R$ has size $|E_R|\leq \alpha |E_H|$, we can call it $\alpha$-optimal.

However, there is a critical issue about the multiplicative approximation factor in our setting. Since we are considering the case when $k>\Omega(n)$, by Moore bound \cite{moore} the classical greedy algorithm will always give a $k$-spanner with girth at least $k+2$ and size at most $n^{1+O(1/n)}=n+O(\log{n})$. Since a trivial lower bound on the size of minimum $k$-spanner $H=(V,E_H)$ of a connected graph is $|E_H|\ge n-1$, any constant multiplicative approximation factor $\alpha=1+\Omega(1)$ will become trivial since the well-known greedy algorithm always outputs a $k$-spanner with size $n+o(n)\leq (1+o(1))|E_H|\leq \alpha |E_H|$. These spanners are usually called ultrasparse spanners \cite{ultrasparse,ultrasparse2}.

Therefore, we must introduce a more refined notion of `approximation'. The first observation is that since different connected components are isolated when considering spanners, we can only focus on connected graphs without loss of generality. 
 For any connected $n$-vertex graph $G=(V,E)$ and its minimum $k$-spanner $H=(V,E_H)$, since $|E_H|\ge n-1$ is a trivial lower bound, we can use $\alpha=\frac{|E_R|-n}{|E_H|-n}$ as a new definition for the approximation factor, where $R=(V,E_R)$ denotes the $k$-spanner of $G$ outputted by some approximation algorithm.
\begin{definition}
For any pair $(n,k)$ and constants $\alpha>1,\beta\ge 0$, if any $n$-vertex connected graph $G=(V,E)$ has a $k$-spanner $H=(V,E_H)$ with girth at least $k+2$ such that $|E_H|-n\leq \alpha(\mathsf{OPT}-n)+\beta$, we call $(n,k)$ an $(\alpha,\beta)$-\emph{approx good} pair. Here $\mathsf{OPT}$ denotes the number of edges in any minimum $k$-spanner of $G$.
\end{definition}
We can observe that the fact $(n,k)$ is $(\alpha,\beta)$-approx good is equivalent to `$(\alpha,\beta)$-approximately universal optimality' of the greedy algorithm. Namely, the greedy algorithm is `$(\alpha,\beta)$-approximately universally optimal' on $(n,k)$ iff for any $n$-vertex graph $G$, the greedy algorithm can output a $(\alpha,\beta)$-approximately minimum $k$-spanner of $G$ on some edge ordering $\sigma$. Using the above definitions, we can try to understand the power of the greedy algorithm when $\frac{1}{3}n\leq k\leq \frac{2}{3}n$. The first bound is the following
\begin{theorem}\label{2approxub}
    There is a deterministic polynomial time algorithm $A$ such that for all sufficiently large $n$ and any $k>\frac{4}{7}n+O(1)$, given any $n$-vertex graph $G$ and its $k$-spanner $H=(V,E_H)\subseteq G$, $A(G,H,k)$ outputs a $k$-spanner $R=(V,E_R)$ of $G$ with girth at least $k+2$. Moreover, the size of 
 $R$ satisfies $|E_R|-n\leq2(E(H)-n)+1$.  
\end{theorem}
As \Cref{gooddef,goodexdef}, we can use \Cref{2approxub} to derive the following
\begin{corollary}\label{2approxubcor1}
        For all sufficiently large $n$ and every $k>\frac{4}{7}n+O(1)$, $(n,k)$ is  $(2,O(1))$-approx good.
\end{corollary}
We also generalize \Cref{2approxub} (but use a less fine-grained way to analyze it) to smaller $k$ but with a larger approximation factor as follows:
\begin{theorem}\label{moreapproxub}
    For $t\in\{1,2,3,4\}$, there is a deterministic polynomial time algorithm $A$ such that for all sufficiently large $n$ and any $k>\frac{4t}{9t-4}n+O(1)$, given any $n$-vertex graph $G$ and its $k$-spanner $H=(V,E_H)\subseteq G$, $A(G,H,k)$ outputs a $k$-spanner $R=(V,E_R)$ of $G$ with girth at least $k+2$. Moreover, the size of 
 $R$ satisfies $|E_R|-n\leq2t^2(E(H)-n)+2t^2$.  
\end{theorem}

When $t=2$, the requirement for $k$ is just $k>\frac{4}{7}n+O(1)$ which is the same as \Cref{2approxub}, but the approximation factor $2t^2=8$ is worse. \Cref{moreapproxub} is useful when $t=\{3,4\}$, which gives us $k>\frac{12}{23}n+O(1)$ and $k>\frac{1}{2}n+O(1)$ bounds respectively.

In fact, \Cref{moreapproxub} is a `bucket-decomposition' generalized version of \Cref{2approxub}, and $t$ is a decomposition parameter that can be set to any positive integer. However, setting $t\ge 5$ cannot derive new approximation bounds so larger $t$ is useless in the current technique. We give a detailed explanation in \Cref{cannotbeyond}. 
\begin{corollary}
        For $t\in\{1,2,3,4\}$, for all sufficiently large $n$ and any $k>\frac{4t}{9t-4}n+O(1)$, $(n,k)$ is $(2t^2,O(1))$-approx good.
\end{corollary}
\begin{remark}
All proofs of our upper (positive) bounds are `partially constructive'. Informally, \Cref{exgoodub,goodub,2approxub,moreapproxub} all provide algorithms to enlarge the girth of an input $k$-spanner $H$ to at least $k+2$, and guarantee the outputted graph is still a $k$-spanner whose size won't become too large compared to the size of $H$. These algorithms can be appended behind some other good algorithms that don't care about girth as `girth enlarger'. For example, suppose there is an algorithm $B$ that constructs very sparse spanners but doesn't guarantee the girth of the outputted spanner, we can feed its output spanner into our algorithms $A$ such that the final spanner outputted by $A$ has large girth, and at the same time it is only `slightly' larger than $B$'s original output.
\end{remark}
\begin{figure}[t!]
\centering 
\includegraphics[width=1.0\textwidth]{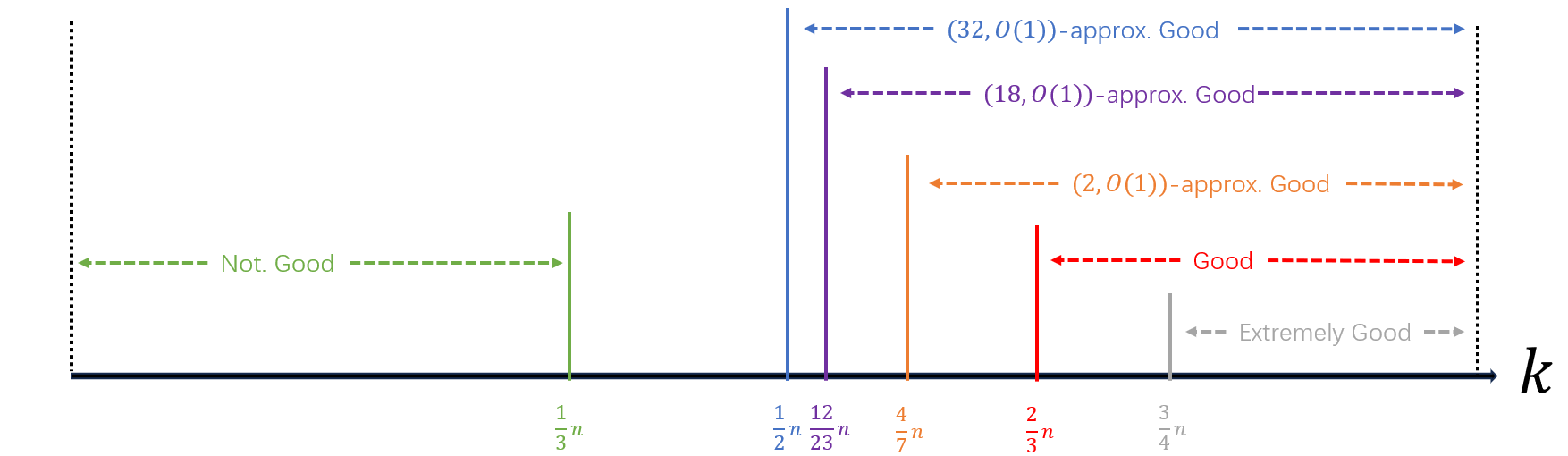} 
\caption{Different Bounds} 
\label{hier} 
\end{figure}
To the best of our knowledge, our proofs are built upon new ideas of analyzing the structure of spanners using the girth information. We hope these ideas could help us to understand the relation between girth and sparse spanners. \Cref{hier} is a summary of our main results.
\subsection{Notations}
For any $n\in\mathbb{N}$, let $[n]=\{1,\dots,n\}$. For any undirected graph $G$, we use $V_G$ and $E_G$ to denote its vertex-set and edge-set unless otherwise stated. For any $(u,v)\in V_G\times V_G$, we use $\dist_{G}(u,v)$ to denote the length of the shortest path between $u,v$ in $G$. For any vertex-subset $U\subseteq V_G$, we use $G[U]$ to denote the induced subgraph of $G$ on $U$. For any two graphs $G,H$, we use $G\backslash H$ to denote the graph $(V_G\backslash V_H,E_G\backslash E_H)$. We can similarly define $G\backslash V_H:=(V_G\backslash V_H,E_G)$ and $G\backslash E_H:=(V_G,E_G\backslash E_H)$. For a simple path $p$ and any two vertices $s,t\in V_p$, we use $p[s,t]$ to denote the sub-path between $s$ and $t$. $p[s,t),p(s,t],p(s,t)$ are similarly defined. We use $SC_G$ to denote an arbitrary smallest cycle of $G$, and $L_G:=|V_{SC_G}|$. 
\subsection{Paper Organization and Overview}
In \Cref{sec:lb}, we give a brief illustration of the negative lower bound regarding the greedy algorithm as stated in \Cref{goodlb}. Then, in \Cref{sec:good} we give detailed proofs of positive upper bounds \Cref{exgoodub} and \Cref{goodub}. Finally, we will prove \Cref{2approxub} and \Cref{moreapproxub} in \Cref{sec:approx}.

Let's give an overview of our proof strategy used in establishing positive upper bounds. Let $H$ be a minimum (in fact arbitrary) $k$-spanner of $G$ and $SC_H$ denote its smallest cycle. Let $L_H=|V_{SC_H}|$. If $L_H\ge k+2$, our girth lower bound has been satisfied and we are done. Otherwise, we want to find a cycle-edge $e\in E_{SC_H}$ and at most $m$ edges $e_1,\dots,e_m\in E_G$, such that we can break the smallest cycle $SC_H$ by removing $e$, and then preserve the property of $k$-spanner by adding these $m$ edges. Formally, let $H'$ denote $\left(H\backslash\{e\}\right)\cup \{e_1,\dots,e_m\}$, we want 
\begin{itemize}
\item $H'$ is a $k$-spanner.
\item $e_1,\dots,e_m$ don't create any new cycle with length at most $k+1$ in $H'$ compared to $H$.
\end{itemize}
If for some pair $(n,k)$, the above condition can be met, then the greedy algorithm is $(m,O(m))$-Approximately universally optimal on $(n,k)$. (See the proof of \Cref{2approxubit} for details.) Therefore, our goal is to prove the above statement in different $m$ and $(n,k)$ settings. First, in \Cref{sec:largecycle} we will show that as long as $L_H>2(n-k)$, we can achieve $m=0$ regardless. It suffices to consider the case when $L_H\leq 2(n-k)$: In \Cref{sec:common}, we give a common structural framework required in all following proofs. Then, we work on different parameter settings as follows
\begin{itemize}
\item In \Cref{sec:ex}, we will prove when $k>\frac{3}{4}n+O(1)$, we can achieve $m=0$, and derive \Cref{exgoodub}.
\item In \Cref{sec:goodpair}, we will prove when $k>\frac{2}{3}n+O(1)$, we can achieve $m=1$, which confirms \Cref{goodub}.
\item In \Cref{sec:2approx}, we will prove when $k>\frac{4}{7}n+O(1)$, we can achieve $m=2$ and derive \Cref{2approxub}. The proof strategy is a generalization of \Cref{sec:goodpair}
\item In \Cref{sec:moreapprox}, we will prove when $k>\frac{12}{23}n+O(1)$, we can achieve $m=18$, and when $k>\frac{1}{2}n+O(1)$, we can achieve $m=32$. These two results follow from a `bucket-decomposition' generalized version of \Cref{2approxub}. 

\end{itemize}
\section{Proof for Lower Bound}\label{sec:lb}
In this section, we give a constructive proof for \Cref{goodlb}. The following theorem implies \Cref{goodlb}.
\begin{theorem}[Restated, \Cref{goodlb}]\label{goodlbp}
For sufficiently large $n$ and $O(1)<k<\frac{1}{3}n-O(1)$, we can construct an $n$-vertex graph $G$ such that all of its minimum $k$-spanners have girth at most $k$.  
\end{theorem}
\begin{figure}[t!]
\centering 
\includegraphics[width=0.4\textwidth]{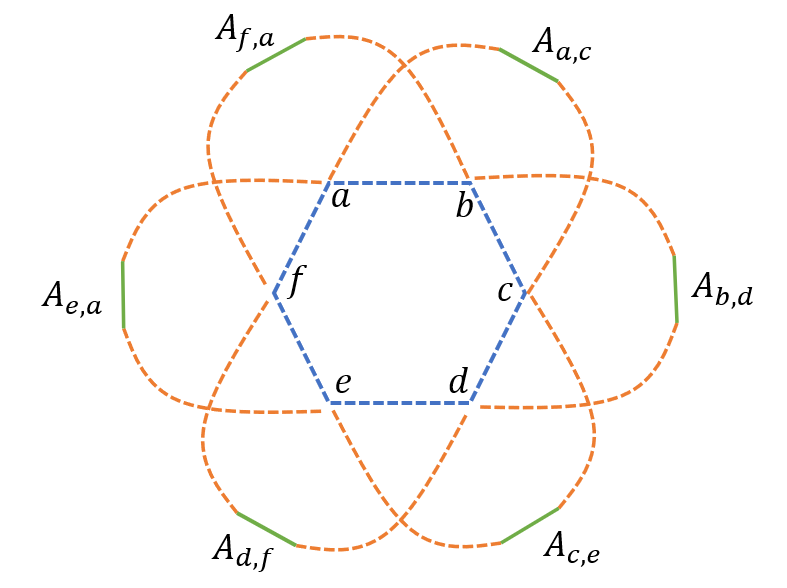} 
\caption{Construction of $G_k$. The dotted blue hexagon represents the $k$-cycle. The dotted orange curves are paths, and the green solid lines are edges. Two orange paths connected by a green edge form an arc. The specific positions of green edges on the arcs don't matter.} 
\label{gk} 
\end{figure}
\begin{proof}
For every sufficiently large $k\ge O(1)$, we can construct a graph $G_k$ as described in \Cref{gk}. The blue $k$-cycle $C=(V,E_C)$ is splited into six segments $C_{a,b},\dots,C_{f,a}$, such that each segment is a path with length $\lfloor \frac{1}{6}k\rfloor$ or $\lfloor \frac{1}{6}k\rfloor+1$. There are also paths outside the cycle composed of green and orange edges. These paths are denoted by $\mathcal{A}=\{A_{a,c},A_{c,e},A_{e,a},A_{b,d},A_{d,f},A_{f,b}\}$ and have lengths $p=2\lfloor\frac{1}{6}k\rfloor+9$ each. We call these paths `arcs'. By simple calculation, $G_k$ has $n'=3k+\Theta(1)$ vertices.

Consider the subgraph $H$ of $G_k$ that consists of only blue and orange edges. Obviously $H$ is a $k$-spanner of $G_k$ with only $n'$ edges. $H$ contains the blue cycle so it has girth $k$. Then we will argue that any $k$-spanner losing any blue edge has than $n'$ edges, which proves that any minimum $k$-spanner of $G_k$ must have girth at most $k$.

We say $A_{a,c}$ covers two cycle-segments $C_{a,b}$ and $C_{b,c}$ and defines the notion `cover' similarly for other arcs and cycle-segments. For any minimum $k$-spanner $H'=(V,E_{H'})$ of $G_k$, we want to prove $H'$ has girth at most $k$. Suppose by contradiction that $H'$ has girth at least $k+1$, then $E_{H'}$ doesn't have some cycle-edge $e\in E_C$, we first identify which cycle-segment $e$ comes from. Here without loss of generality, we suppose $e\in E_{C_{a,b}}$. Then, for any edge $e'=(u,v)\in E_{A_{a,c}}$, if $e'$ is also not in $E_{H'}$, we have
\begin{align}
\dist_{H'}(u,v)&\ge \dist_{G\backslash\{e,e'\}}(u,v)\ge |E_{A_{a,c}}|-1+|E_{C_{c,d}}|+|E_{C_{d,e}}|+|E_{C_{e,f}}|+|E_{C_{f,a}}|\\
&\ge (2\lfloor\frac{1}{6}k\rfloor+9)-1+4\lfloor\frac{1}{6}k\rfloor\ge k+1\label{eq1}
\end{align}
The second inequality above can be derived from a case analysis of possible routes. We omit details here. 

\Cref{eq1} shows that, if $A_{a,c}\nsubseteq H'$, then $H'$ is not a $k$-spanner, which is a contradiction. Therefore, there must be $A_{a,c}\subseteq H'$, the similar analysis also derives $A_{f,b}\subseteq H'$. 

\begin{claim}\label{cllb}
For any $k$-spanner $H'$ of $G_k$ such that $C\nsubseteq H'$, there is $|E_{H'}|\ge n'+1$
\end{claim}
\begin{proof}
Consider the six cycle-segments $C_{a,b},C_{b,c},C_{c,d},C_{d,e},C_{e,f},C_{f,a}$. Suppose $H'$ loses cycle-edges from $t\ge 1$ cycle-segments. First, since $H'$ is a $k$-spanner, it has to preserve connectness, so it loses at most one edge from each cycle-segments. Therefore, $E_{H'}$ contains $k-t$ cycle-edges in $E_C$. Then consider each arc $A\in\mathcal{A}$, if $A$ doesn't cover any cycle-segment that $H'$ loses, by connectedness consideration $H'$ contains at least $|E_A|-1$ edges from $A$. Otherwise, by the previous discussion, $H'$ must contain the whole arc $A$ with $|E_A|$ edges. Let $t'$ denote the number of arcs that cover at least one cycle-segment lost by $H'$, we have $|E_{H'}|\ge k-t+6p-(6-t')$. If $t=6$, $H'$ will be disconnected so there must be $1\leq t\leq 5$. In these cases, by some case analysis we have $t'>t$. Since $n'=k+6p-6$, it follows that $|E_{H'}|\ge n'-t+t'\ge n'+1$.
\end{proof}
By \Cref{cllb}, any $k$-spanner $H'$ of $G_k$ with girth at least $k+1$ will have more edges than the $k$-spanner $H$ with $n'$ edges. Therefore, $H'$ cannot be minimum $k$-spanner of $G_k$. Let $v_0\in V_{G_k}$ be an arbitrary vertex in $G_k$. For any $n> n'=3k+\Theta(1)$, we can construct $G=(V,E)$ as follows
\begin{compactitem}
\item $V:=V_{G_k}\cup\{v_{1},\dots,v_{n-n'}\}$, where $v_{1},\dots,v_{n-n'}$ are brand new vertices.
\item $E:=E_{G_k}\cup\{(v_i,v_{i+1})\colon i\in\{0\}\cup [n-n'-1]\}$
\end{compactitem}
We can observe that $G$ has $n$ vertices and doesn't have any minimum $k$-spanner with girth larger than $k$, which proves this theorem.
\end{proof}
\section{Proofs for Upper Bounds}\label{sec:good}
\begin{figure}[t!]
\centering 
\includegraphics[width=0.3\textwidth]{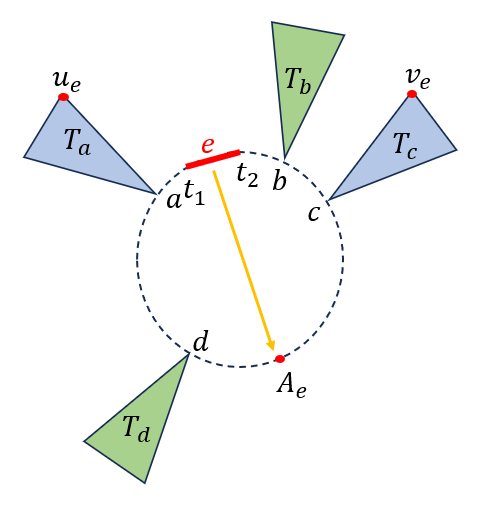} 
\caption{The whole graph denotes $K$ and the middle circle denotes $SC_H$.} 
\label{largecyclepic} 
\end{figure}
In this section, we give proofs for \Cref{exgoodub} and \Cref{goodub}. When $k$ is large enough, we will argue that given any $k$-spanner whose smallest cycle has length at most $k+1$, we can algorithmically break the cycle, and then add zero/one edge such that after modification we don't create new cycles smaller than $k+2$, and it is still a $k$-spanner. We will discuss two cases when the smallest cycle is larger than $2(n-k)$ or smaller than that.

For any graph $G$, recall $SC_G$ denotes its smallest cycle (If there are multiple smallest cycles, we break the tie arbitrarily). We will also assign an orientation for $SC_G$ for ease of analysis. For any cycle $C$ with orientation and two vertices $a,b\in V_C$, we use $C[a,b]$ to denote the cycle-segment(path) on $C$ from $a$ to $b$ through the assigned orientation.

\subsection{When the Smallest Cycle Has Length $2(n-k)\leq L_H\leq k+1$ }\label{sec:largecycle} First we aim to show the following lemma that argues for the case when the smallest cycle isn't too small.
\begin{lemma}\label{largecycle}
Given any $n$-vertex graph $G$ and its $k$-spanner $H\subseteq G$. Let $L_H=|V_{SC_H}|$. If $2(n-k)\leq L_H\leq k+1$, we can always find an edge $e\in E_{SC_H}$, such that $H\backslash\{e\}$ is still a $k$-spanner.
\end{lemma}
\begin{proof}
We need several steps to prove \Cref{largecycle}. Without loss of generality, we can assume $H$ is a connected graph. First, by a single multi-source BFS on $H$, we can compute for each vertex $v\in V_H$, its closest cycle-vertex $b_v\in V_{SC_H}$ such that $d_v:=\dist_{H}(v,b_v)$ is minimized over all $b_v\in V_{SC_H}$. If there are multiple choices achieving minimum, we break the tie arbitrarily. Then, for any $u\in V_{SC_H}$, let $B_u:=\{v\in V_H\colon b_v=u\}$, we run a BFS on $H[B_u]$ from source $u$, and use $T_u=(B_u,E_{T_u})$ to denote the generated BFS tree. We can construct a subgraph $K:=\left(\bigcup_{u\in V_{SC_H}}T_u\right)\cup SC_H$ of $H$ such that $K$ contains exactly these BFS trees and $SC_H$. \Cref{largecyclepic} is an example. We know $|E_K|=n=|V_G|$. Then, we have the following observation

\begin{fact}\label{half}
For any two vertices $a,b\in V_{SC_H}$ such that $|E_{SC_H[a,b]}|\leq \frac{L_H}{2}$, we have $|E_{T_a}|+|E_{T_b}|+|E_{SC_H[a,b]}|\leq k$
\end{fact}
\begin{proof}
We recall the notation that $SC_H[a,b]$ denotes the cycle-segment(path) from $a$ to $b$ on the cycle $SC_H$. Here we assign an arbitrary orientation to $SC_H$. We derive \Cref{half} from the following
\begin{align*}
|E_{T_a}|+|E_{T_b}|+|E_{SC_H[a,b]}|&\leq|E_K|-|E_{SC_H}|+|E_{SC_H[a,b]}|\leq n-L_H+L_H/2\\
&=n-L_H/2\leq n-(n-k)=k
\end{align*}
\end{proof}
Suppose by contradiction that for any $e\in E_{SC_H}$, $H\backslash\{e\}$ isn't $k$-spanner. Then, we know for any $e\in E_{SC_H}$, there must be $u_e,v_e,a,c$ such that $a\neq c\in V_{SC_H}, e\in E_{SC_H[a,c]}, u_e\in V_{T_a}, v_e\in V_{T_c}$ and the following inequality holds
\begin{equation}\label{eq:contra}
|E_{T_a}|+|E_{T_c}|+|E_{SC_H[c,a]}|\ge \dist_{K\backslash \{e\}}(u_e,v_e)\ge \dist_{H\backslash\{e\}}(u_e,v_e)>k.
\end{equation}
\Cref{largecyclepic} is an example. From \Cref{half}, we know $|E_{SC_H[c,a]}|>L_H/2$ and $|E_{SC_H[a,c]}|=L_H-|E_{SC_H[c,a]}|< L_H/2$.

Let $e=(t_1,t_2)$ where $t_2$ is next to $t_1$ according to the orientation of $SC_H$. Let $A_e\in V_{SC_H}$ denote the antipode of $e$ on the cycle $SC_H$. That is, $A_e$ is the foremost vertex on $SC_H[t_2,t_1]$ that has distance at least $\lfloor\frac{L_H-1}{2}\rfloor$ from $t_2$. Then we meet the second crucial observation.
\begin{fact}\label{split}
Neither $V_{SC_H[t_2,A_e]}$ nor $V_{SC_H[A_e,t_1]}$ contains the two vertices $a,c$ simultaneously.
\end{fact}
\begin{proof}
It's obvious that both $|E_{SC_H[t_2,A_e]}|$ and $|E_{SC_H[A_e,t_1]}|$ are at most $\frac{L_H}{2}$. Suppose by contradiction that $a,c$ are both from one of them, say $a,c\in V_{SC_H[t_2,A_e]}$ without loss of generality. Since we know $e\notin E_{SC_H[c,a]}$, we must have $SC_H[c,a]\subseteq SC_H[t_2,A_e]$ and $|E_{SC_H[c,a]}|\leq |E_{SC_H[t_2,A_e]}|\leq\frac{L_H}{2}$. Then by \Cref{half} we know $|E_{T_a}|+|E_{T_c}|+|E_{SC_H[c,a]}|\leq k$, which violates \Cref{eq:contra}.
\end{proof}
Let $m(K)$ denote the number of non-trivial trees in $K$ that originated from the cycle. That is, $m(K):=\#\{a:a\in V_{SC_H}\wedge |E_{T_a}|>0\}$. For example, the graph $K$ in \Cref{largecyclepic} has $m(K)=4$. Let $Tr(K)$ denote the set $\{a:a\in V_{SC_H}\wedge |E_{T_a}|>0\}$.
\begin{fact}\label{treenumber}
$m(K)\ge 2$
\end{fact}
\begin{proof}
Since $L_H\leq k+1\leq n-1$, there must be some vertex $a\notin V_{SC_H}$, so $m(K)=|Tr(K)|>0$.

Suppose by contradiction that $m(K)=|Tr(K)|=1$. Let $b\in Tr(K)$, we can always find an edge $e\in E_{SC_H}$ such that $A_e=b$. By previous discussion, we can find $u_e,v_e,a,c$ satisfying certain conditions stated above. By \Cref{split}, since $b=A_e$, we know neither $a$ nor $c$ equals to $b$, so $a,c\notin Tr(K)$ and $|E_{T_a}|=|E_{T_c}|=0$. This implies $|E_{T_a}|+|E_{T_c}|+|E_{SC_H[c,a]}|\leq L_H-1\leq k$ since $a\neq c$, which violates \Cref{eq:contra}. Therefore, we must have $m(K)=|Tr(K)|\ge 2$.
\end{proof}
For any two vertices $a,c\in V_{SC_H}$ on the cycle and $e\in E_{SC_H[a,c]}$, if the inequality $|E_{T_a}|+|E_{T_c}|+|E_{SC_H[c,a]}|>k$ holds, we say the pair $(a,c)$ covers $e$. From \Cref{eq:contra}, we know each cycle-edge $e\in E_{SC_H}$ is covered by some pair. Our next step is to analyze the coverage relation after removing some tree. For any $b\in V_{SC_H}$, we use $V'_{T_b}:=V_{T_b}\backslash\{b\}$ to denote the set of all vertices of tree $T_b$ except for its root $b$. We use the following lemma to clarify the structure of $K$.
\begin{lemma}\label{trim}
For any cycle-edge $e\in E_{SC_H}$ covered by $(a,c)$. If there exists $b\in V_{SC_H[a,c]}\backslash\{a,c\}$ such that (1). $|E_{T_b}|>0$. (2). neither $(b,c)$ nor $(a,b)$ covers $e$. Then, $K'=K[V_K\backslash V'_{T_b}]$ is still a graph such that all cycle-edges are covered by some pair.
\end{lemma}
\begin{proof}
An example is shown in \Cref{largecyclepic}. Let $e=(t_1,t_2)$ where $t_2$ is next to $t_1$ by the orientation of $SC_H$. Without loss of generality, suppose $b\in V_{SC_H[t_2,c]}$. Let $K'=K[V_K\backslash V'_{T_b}]$ after removing tree $T_b$ from $K$. Since $e$ is covered by $(a,c)$ but not by $(a,b)$, we know the following:
\begin{align}
(1)\ &|E_{T_a}|+|E_{T_c}|+|E_{SC_H[c,a]}|>k\\
(2)\ &|E_{T_a}|+|E_{T_b}|+|E_{SC_H[b,a]}|\leq k\\
(1)-(2)\ \Rightarrow & |E_{T_b}|+|E_{SC_H[b,c]}|<|E_{T_c}|\label{eqreduce}
\end{align}

Moreover, by \Cref{half} we know $|E_{SC_H[c,a]}|>\frac{L_H}{2}$. Since all cycle-edges $e'\in E_{SC_H[a,c]}$ are covered by $(a,c)$ in $K'$. we only need to show that all other cycle-edges $e'\in E_{SC_H[c,a]}$ are still covered by some pair in $K'$. For any  $e'\in E_{SC_H[c,a]}$, suppose by contradiction it isn't covered by any pair in $K'$. Since we know $e'$ is covered by some pair in $K$, there are two cases:
\begin{itemize}
\item[(1)] For some $d\in V_{SC_H}$, there is $(d,b)$ covers $e'$ in $K$. If so, we know $|E_{T_d}|+|E_{T_b}|+|E_{SC_H[b,d]}|>k$ and $e'\in E_{SC_H[d,b]}$. By \Cref{half}, we know $|E_{SC_H[b,d]}|>\frac{L_H}{2}$. Since $|E_{SC_H[b,c]}|\leq |E_{SC_H[a,c]}|=L_H-|E_{SC_H[c,a]}|\leq \frac{L_H}{2}$, we have $d\notin V_{SC_H[b,c]}$. It follows that
\begin{align*}
&|E_{T_c}|+|E_{T_d}|+|E_{SC_H[c,d]}|\\
>&|E_{T_b}|+|E_{SC_H[b,c]}|+|E_{T_d}|+|E_{SC_H[c,d]}|\\
\ge& |E_{T_d}|+|E_{T_b}|+|E_{SC_H[b,d]}|>k
\end{align*}
The second inequality is given by \Cref{eqreduce}. Moreover, $e'\in E_{SC_H[d,b]}\subseteq E_{SC_H[d,c]}$. We can conclude $e'$ is covered by $(d,c)$ in $K'$, which is a contradiction if we assume $e'$ is not covered by any pair in $K'$.

\item[(2)] For some $d\in V_{SC_H}$, there is $(b,d)$ covers $e'$ in $K$. If so, we know $|E_{T_d}|+|E_{T_b}|+|E_{SC_H[d,b]}|>k$ and $e'\in \left(E_{SC_H[b,d]}\cap E_{SC_H[c,a]}\right)$, then there must be $d\in V_{SC_H[c,b]}$. We can derive that
\begin{align*}
&|E_{T_c}|+|E_{T_d}|+|E_{SC_H[d,c]}|\\
>&|E_{T_b}|+|E_{SC_H[b,c]}|+|E_{T_d}|+|E_{SC_H[d,c]}|\\
\ge& |E_{T_d}|+|E_{T_b}|+|E_{SC_H[d,b]}|>k
\end{align*}
The second inequality is given by \Cref{eqreduce}. Moreover, $e'\in \left(E_{SC_H[b,d]}\cap E_{SC_H[c,a]}\right)\subseteq E_{SC_H[c,d]}$. We can conclude $e'$ is covered by $(c,d)$ in $K'$, which is a contradiction again.
\end{itemize}
Therefore, $K'=K[V_K\backslash V'_{T_b}]$ is still a graph in that all cycle-edges are covered by some pair.
\end{proof}
We repeatedly invoke \Cref{trim} to reduce the size of $K$. Since each time if we can find such an $e$ satisfying conditions of \Cref{trim} in $K$, the trimmed graph $K'=K[V_K\backslash V'_{T_b}]$ will lose a non-trivial tree $T_b$ where $|E_{T_b}|>0$, we have $m(K')=m(K)-1$. By \Cref{treenumber}, after at most $O(n)$ iterations, we will get a $K'$ such that $n'=|V_{K'}|<n$, all of its cycle-edges are covered by some pair in $K'$, and we cannot use \Cref{trim} to make it smaller. It means that for each cycle-edge $e\in E_{SC_H}$, we can find $(a,b)$ that covers $e$ in $K'$ and there isn't any other cycle-vertex $c\in \left(V_{SC_H[a,b]}\cap Tr(K')\right)$ with non-trivial tree $T_c$ except for $a$ and $b$. The remaining concern is that either $T_a$ or $T_b$ may be a trivial tree, which is bad for our analysis. However, we have the following structural argument on $K'$ to help us exclude this case.
\begin{figure}[t!]
\centering 
\includegraphics[width=0.3\textwidth]{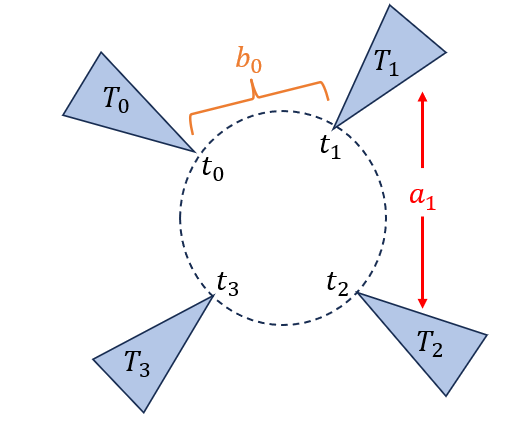} 
\caption{Definitions of $a_i,b_i$} 
\label{def} 
\end{figure}
\begin{claim}\label{finalstruc}
For each $e\in E_{SC_H}$, there exists $(a,b)$ covering $e$ in $K'$ such that (1). $a,b\in Tr(K')$. (2). There isn't any $c\in \left(V_{SC_H[a,b]}\cap Tr(K')\right)$ except for $a$ or $b$.
\end{claim}
\begin{proof}
Fix any $e=(u,v)\in E_{SC_H}$, where $v$ is next to $u$ according to the orientation. Let $s$ denote the first vertex in $Tr(K')$ after $v$ according to the orientation, and $t$ denote the first vertex in $Tr(K')$ before $u$. From a similar argument as in \Cref{treenumber}, we know $m(K')\ge 2$ and $s\neq t$. There are two possibilities:
\begin{itemize}
\item[(1)] If $|E_{T_s}|+|E_{T_t}|\ge|E_{SC_H[t,s]}|$, then there exists an edge $e'=(t_1,t_2)\in E_{SC_H[t,s]}$ where $t_2$ is next to $t_1$ in the orientation, such that for any $a'\in V_{SC_H[t,t_1]},b'\in V_{SC_H[t_2,s]}$, both $|E_{T_{a'}}|+|E_{SC_H[t,a']}|\leq |E_{T_{t}}|$ and $|E_{T_{b'}}|+|E_{SC_H[b',s]}|\leq |E_{T_{s}}|$ hold.  From property of $K'$ discussed above and \Cref{trim}, there exists such a pair of $(a',b')$ such that $SC_H[a',b']\subseteq SC_H[t,s]$ and $(a',b')$ covers $e'$. Therefore, $|E_{T_{a'}}|+|E_{T_{b'}}|+|E_{SC_H[b',a']}|>k$. Moreover, we have
\begin{align*}
&|E_{T_t}|+|E_{T_s}|+|E_{SC_H[s,t]}|\\
\ge&|E_{T_{a'}}|+|E_{SC_H[t,a']}|+|E_{T_{b'}}|+|E_{SC_H[b',s]}|+|E_{SC_H[s,t]}|\\
\ge&|E_{T_{a'}}|+|E_{T_{b'}}|+|E_{SC_H[b',a']}|>k
\end{align*}
Since $e\in E_{SC_H[t,s]}$, $e$ must be covered by $(t,s)$ and both conditions in the statement hold. We are done.

\item[(2)]\label{item:second} If $|E_{T_s}|+|E_{T_t}|<|E_{SC_H[t,s]}|$, then there exists an edge $e'=(t_1,t_2)\in E_{SC_H[t,s]}$ where $t_2$ is next to $t_1$ in the orientation, such that for any $a'\in V_{SC_H[t,t_1]},b'\in v_{SC_H[t_2,s]}$, we get $|E_{T_{a'}}|+|E_{SC_H[t,a']}|\leq |E_{SC_H[t,t_1]}|$ and $|E_{T_{b'}}|+|E_{SC_H[b',s]}|\leq |E_{SC_H[t_2,s]}|$.
From property of $K'$ discussed above and \Cref{trim}, there exists a pair $(a',b')$ such that $SC_H[a',b']\subseteq SC_H[t,s]$ and $(a',b')$ covers $e'$. Therefore, $|E_{T_{a'}}|+|E_{T_{b'}}|+|E_{SC_H[b',a']}|>k$. However, we have
\begin{align*}
&|E_{T_{a'}}|+|E_{T_{b'}}|+|E_{SC_H[b',a']}|\\
=&|E_{T_{a'}}|+|E_{SC_H[t,a']}|+|E_{T_{b'}}|+|E_{SC_H[b',s]}|+|E_{SC_H[s,t]}|\\
\leq&|E_{SC_H[t,t_1]}|+|E_{SC_H[t_2,s]}|+|E_{SC_H[s,t]}|\leq L_H-1\leq k
\end{align*}
This is a contradiction, so the second case is in fact impossible.
\end{itemize}
\end{proof}
Finally, For each cycle-edge $e=(u,v)\in E_{SC_H}$ where $v$ is next to $u$ in the orientation, we use $s_e,t_e\in Tr(K')$ to denote the first vertices in $Tr(K')$ after $v$ and before 
$u$ respectively. By \Cref{finalstruc} we know for every $e\in E_{SC_H}$, $e$ is covered by $(t_e,s_e)$ in $K'$. Let $Tr(K')=\{t_0,\dots,t_{m(K')-1}\}$ denote all elements in $Tr(K')$ by the oriented order. We will split the cycle $SC_H$ into $m(K')$ segments. Consider any $i\in\{0\}\cup[m(K')-1]$, let $nxt(i):=i+1\mod{m(K')}$, we define $b_i:=|E_{SC_H[t_i,t_{nxt(i)}]}|$ and $a_i:=|E_{T_{t_i}}|+|E_{T_{t_{nxt(i)}}}|$. \Cref{def} is an example of these definitions.

Since $2\leq m(K')\leq L_H$, we have the following three basic equations:
\begin{align}
(1)\ &\sum_{i=0}^{m(K')-1}b_i=L_H\\
(2)\ &\sum_{i=0}^{m(K')-1}a_i=2(n'-L_H)\\
\frac{(2)-(1)}{m(K')}\Rightarrow &\ \frac{\sum_{i=0}^{m(K')-1}a_i-b_i}{m(K')}\leq\max\left(-\frac{1}{L_H},\frac{2n'-3L_H}{2}\right)=\max\left(-\frac{1}{L_H},n'-\frac{3}{2}L_H\right) \label{average}
\end{align}
By an average argument from \Cref{average}, there exists $i\in\{0\}\cup[m(K')-1]$ such that $a_i-b_i\leq \max\left(-1,n'-\frac{3}{2}L_H\right)$. If $a_i-b_i\leq -1$, we fall into the second case of \Cref{finalstruc}, which was shown to be impossible. Therefore, there must be $a_i-b_i\leq n'-\frac{3}{2}L_H$, we can derive
\begin{align*}
&|E_{T_{t_i}}|+|E_{T_{nxt(i)}}|+|E_{SC_H[t_{nxt(i)},t_i]}|\\
=&a_i+L_H-b_i\leq n'-\frac{1}{2}L_H\leq n'-(n-k)\leq k
\end{align*}

Therefore, any cycle-edge $e\in E_{SC_H[t_i,t_{nxt(i)}]}$ is not covered by $(t_e,s_e)=(t_i,nxt(i))$, which contradicts the property of $K'$ discussed before and \Cref{finalstruc}. As a result, it implies the initial assumption that `for any cycle-edge $e\in E_{SC_H}$, $H\backslash\{e\}$ isn't a $k$-spanner' doesn't hold. We conclude that we can find a cycle-edge $e\in E_{SC_H}$
 such that $H\backslash\{e\}$ is still a $k$-spanner.
\end{proof}

By \Cref{largecycle}, given any $k$-spanner $H$ of $G$, if $H$ has girth at least $\min\left(k+2,2(n-k)\right)$, we can algorithmically get a $k$-spanner $H'\subseteq H$ with girth at least $k+2$.  Now it suffices to only consider the case when the input $k$-spanner $H$ of $G$ has the smallest cycle $SC_H$ of length smaller than $2(n-k)$.
\subsection{Some Common Proof Techniques in Upper Bounds}\label{sec:common} In this paragraph we will introduce several common notions that will be used in our proofs of \Cref{exgoodub,goodub,2approxub,moreapproxub}.
\begin{definition}\label{dangerdef}
For any $n$-vertex graph $G$ and its $k$-spanner $H$, if for an edge $b=(s,t)\in E_G$ and a cycle-edge $e\in E_{SC_H}$, we have $\dist_{H\backslash\{e\}}(s,t)>k$, we say $(s,t)$ is an endangered pair and call $e$ a \emph{danger} for $(s,t)$.
\end{definition}

Given any $k$-spanner $H$ of $G$, we aim to transform $H$ into a $k$-spanner with girth at least $k+2$. From the previous paragraph, we can assume $|E_{SC_H}|<2(n-k)$. If there exists a cycle-edge $e\in E_{SC_H}$ that isn't a danger for any vertex-pair, we just remove $e$ and $H\backslash\{e\}$ is still a $k$-spanner of $G$ with fewer edges and its girth won't become smaller. By iteratively doing so, we finally achieve a $k$-spanner $H$ whose any cycle-edges $e\in E_{SC_H}$ is a danger for some pair.

For any $s,t\in V_{G}$, let $P_{s,t}$ denote the set of the shortest paths between $s$ and $t$ on $H$. For any path $p\in P_{s,t}$, we denote its oriented edge set as $E_p:=(e_0=(v_0=s,v_1),\dots,e_{m-1}=(v_{m-1},v_m=t))$ and its sub-path from 
$v_i$ to $v_j$ as $p[v_i,v_j],0\leq i\leq j\leq m$, where $m=\dist_H(s,t)$. Moreover, let $C_p=\{e_i\in E_p:e_i\in E_{SC_H}\}$ denote the set of those cycle-edges in $p$, we observe the following fact.
\begin{fact}\label{depen}
If $e$ is a danger for $(s,t)$, then for any $p\in P_{s,t}$, $e\in p$.
\end{fact}
\begin{proof}
Suppose it's not the case, we can find $p\in P_{s,t}$ such that $p\subseteq H\backslash\{e\}$. Since $H$ is a $k$-spanner we have $|E_p|=\dist_H(s,t)\leq k$, which means $\dist_{H\backslash\{e\}}(s,t)\leq |E_{p}|\leq k$, which contradicts the fact that $e$ is a danger for $(s,t)$ by \Cref{dangerdef}.
\end{proof}

For any $(s,t)$ and its danger $e$, let $\mathsf{cd}(s,t,e):=\mathsf{argmax}_{p\in P_{s,t}}\{|E_{C_p}|\}$ denote one of the shortest path between $s$ and $t$ that has the longest common parts with the cycle $SC_H$ (We can break the tie arbitrarily). Let $\mathsf{Ccd}(s,t,e):=C_{cd(s,t,e)}$ denote its cycle parts, From \Cref{depen}, we know $|\mathsf{Ccd}(s,t,e)|>0$. We give the next observation.

\begin{figure}[t!]
\centering 
\includegraphics[width=0.4\textwidth]{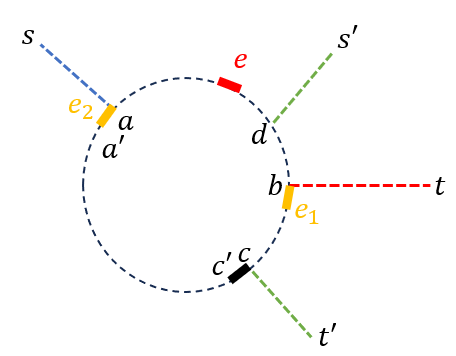} 
\caption{An illustration for our analysis. On this subgraph of $H$, the path between $s$ and $t$ using $SC_H[a,b]$ is $p$, and the path between $s'$ and $t'$ using $SC_H[d,c]$ is $p'$. The orientation of the cycle $SC_H$ is clockwise.} 
\label{upg} 
\end{figure}

\begin{claim}\label{conse}
If $e$ is a danger for $(s,t)$, then $\mathsf{Ccd}(s,t,e)$ forms a consecutive segment both in $E_{\mathsf{cd}(s,t,e)}$ and $E_{SC_H}$. Moreover, $e\in \Ccd(s,t,e)$.
\end{claim}
\begin{proof}
If $|\Ccd(s,t,e)|=1$, by \Cref{depen} we are done. Otherwise according to the order appeared in $E_{\cd(s,t,e)}=(e_0=(v_0=s,v_1),\dots,e_{m-1}=(v_{m-1},v_m=t))$, we can choose $b_1=(v_i,v_{i+1}),b_2=(v_j,v_{j+1})\in \Ccd(s,t,e), i<j$ to denote the first and last edges in $\Ccd(s,t,e)$. Suppose by contradiction that $\Ccd(s,t,e)$ doesn't form a consecutive segment on $E_{\cd(s,t,e)}$, we know $|\Ccd(s,t,e)|<|E_{p[i,j+1]}|=j-i+1$ and there exists $i<k<j$ such that $e_k\notin E_{SC_H}$. Since $p[i,j+1]\in P_{v_i,v_{j+1}}$, we know $\dist_H(v_i,v_{j+1})=|E_{p[i,j+1]}|=j-i+1$. Now consider the cycle-segment $SC_H[v_{j+1},v_{i}]$, there is a cycle $C=SC_H[v_{j+1},v_{i}]\cup p[i,j+1]$ that doesn't equal to $SC_H$ since $e_k\in E_C\backslash E_{SC_H}$. However, we know $SC_H=SC_H[v_{j+1},v_{i}]\cup SC_H[v_i,v_{j+1}]$ is one of the smallest cycle in $H$. Therefore, we must have $|E_{SC_H[v_i,v_{j+1}]}|\leq |E_{p[i,j+1]}|$. Since they are both paths from $v_i$ to $v_{j+1}$ and $p[i,j+1]\in P_{v_i,v_{j+1}}$, we actually have $SC_H[v_i,v_{j+1}]\in P_{v_i,v_{j+1}}$. Therefore, $p'=p[0,i]\cup SC_H[v_i,v_{j+1}]\cup p[j+1,m]$ is also a shortest path between $s$ and $t$ and $|C_{p'}|=|E_{SC_H[v_i,v_{j+1}]}|=\dist_H(v_i,v_{j+1})=j-i+1>|\Ccd(s,t,e)|$. This contradicts our definition of $\Ccd(s,t,e)$ since it implies $|C_{p'}|>|\Ccd(s,t,e)|=|C_{\cd(s,t,e)}|$ but $\cd(s,t,e)$ should have had longer cycle part than $p'$. Therefore, our assumption doesn't hold and $\Ccd(s,t,e)$ must form a consecutive segment in both $E_{\cd(s,t,e)}$ and $E_{SC_H}$. 
\end{proof}

Next, we choose the path $p$ which is the shortest path between some arbitrary endangered pair, and the choice should maximize $|C_{p}|$. Namely, we define $(s,t,e):=\mathsf{argmax}_{s,t,e}\{|\Ccd(s,t,e)|\}$ and $p=\cd(s,t,e)$ (We can break ties arbitrarily). Then, we fix $p$ and define the subgraph $K:=p\cup SC_H\subseteq H$. See \Cref{upg} for an example. Let $a,b$ be the first and last cycle-vertices on $p$ and use the direction from $a$ to $b$ as the orientation of $SC_H$. From \Cref{conse}, we know $e\in E_{SC_H[a,b]}$ and $p=p[s,a]\cup SC_H[a,b]\cup p[b,t]$. Let $a'$ denote the previous vertex of $a$ on $SC_H$ and $b'$ the next vertex of $b$ on $SC_H$, we focus on two cycle-edges $e_1=(b,b')$ and $e_2=(a',a)$ on the cycle. Without loss of generality, we can assume $|E_{p[b,t]}|\ge |E_{p[s,a]}|$. 

In all the following proofs, we assume the orientation of the cycle $SC_H$ is the same as path $p$ oriented from $a$ to $b$. In all following figures, this orientation is described as clockwise as in \Cref{upg}.
\subsection{Upper Bound for Extremely Good Pairs}\label{sec:ex}
In this paragraph, we aim to prove the upper bound for extremely good pairs \Cref{exgoodub}. The following lemma is the central part when $|E_{SC_H}|=L_H<2(n-k)$.
\begin{lemma}\label{exsmallcycle}
If $k>\frac{3}{4}n+O(1)$, and $H$ is the $k$-spanner discussed in \Cref{sec:common} with $|E_{SC_H}|< 2(n-k)$. Then at least one of the $H\backslash\{e_1\}$ and $H\backslash\{e_2\}$ is still a $k$-spanner, where $e_1,e_2$ is also defined in \Cref{sec:common}.
\end{lemma}
\begin{proof}
If $H\backslash\{e_1\}$ is a $k$-spanner, we are done. In the following, let's assume $H\backslash\{e_1\}$ isn't a $k$-spanner. Then $e_1$ must be a danger for some pair $(s',t')$. Let $p_1=\cd(s',t',e_1)$. From \Cref{depen}, we know $e_1=(b,b')\in E_{p_1}$. Without loss of generality, we can assume $b$ appears earlier than $b'$ on $p_1$. Let $d,c$ denote the first and last cycle-vertices in $\Ccd(s',t',e_1)$, by \Cref{conse}, we know $p_1=p_1[s',d]\cup SC_H[d,c]\cup p_1[c,t']$. Then let's consider the relationship between $(s,t)$ and $(s',t')$.

Since $SC_H[a,b]\subseteq p$, $SC_H[a,b]$ is a shortest path between two cycle-vertices $a,b\in V_{SC_H}$ and $|E_{SC_H[a,b]}|\leq \frac{L_H}{2}$.

\begin{claim}\label{placee1}
$d\in V_{SC_H[a,b]}, c\in V_{SC_H[b',a]}$.
\end{claim}
\begin{proof}
Let's eliminate all other potential possibilities. There are two of them:
\begin{itemize}
\item[(1)] Suppose by contradiction $d\in V_{SC_H[b',a']}$, then $SC_H[a,b]\subseteq SC_H[d,b]$. Since $e_1=(b,b')\in SC_H[d,c]$, we have $SC_H[d,c]=SC_H[d,b]\cup SC_H[b,c]$. Since $c$ appears after $b'$ in $SC_H[d,c]$, we have $e_1\in E_{SC_H[b,c]}$ and $SC_H[a,b']\subsetneq SC_H[d,b]\cup SC_H[b,c]=SC_H[d,c]$, which implies $\Ccd(s,t,e)\subsetneq \Ccd(s',t',e_1)$. This contradicts the definition that $|\Ccd(s,t,e)|$ is one of the longest common cycle-segment among the choices of endangered pairs $(s,t)$ and its dangers $e$. Therefore this possibility is impossible and $d\in V_{SC_H[a,b]}$.

\item[(2)] Suppose by contradiction $c\in V_{SC_H[a,b]}$, Since we've shown $d\in V_{SC_H[a,b]}$ and $b'\in V_{SC_H[d,c]}$, we must have $SC_H[b,a]\subseteq SC_H[d,c]$. It implies $SC_H[c,d]\subseteq SC_H[a,b]$ and $|E_{SC_H[c,d]}|\leq |E_{SC_H[a,b]}|\leq \frac{L_H}{2}\leq |E_{SC_H[b,a]}|\leq |E_{SC_H[d,c]}|$. At the same time, $SC_H[d,c]\subseteq p_1\in P_{s',t'}$, which means $SC_H[d,c]\in P_{d,c}$. Since $|E_{SC_H[c,d]}|\leq |E_{SC_H[d,c]}|$, we derive $SC_H[c,d]\in P_{d,c}$ also. However, in this case $p':=p_1[s',d]\cup SC_H[c,d]\cup p_1[c,t']\in P_{s',t'}$ will be a shortest path between $s'$ and $t'$ in $H$ but $e_1\notin p'$, which contradicts \Cref{depen} since $e_1$ is a danger for $(s',t')$.
\end{itemize}
\end{proof}

By \Cref{placee1}, we can confirm our graph drawing as shown in \Cref{upg}. However, $p[s,a],p_1[s',d],p[b,t]$ and $p_1[c,t']$ may overlap. In the next step, we should separate them by a more fine-grained analysis.

\begin{claim}\label{innernon}
$V_{p[s,a]}$ doesn't overlap with $V_{p_1[s',d]}$. 
\end{claim}
\begin{proof}
Suppose by contradiction, there is some vertex $u$ both in $V_{p[s,a]}$ and $V_{p_1[s',d]}$. Then, $p[u,d]\in P_{u,d}$ must be the shortest path between $u$ and $d$. Then, we know $p':=p_1[s',u]\cup p[u,d]\cup p_1[d,c]\cup p_1[c,t']\in P_{s',t'}$ is a shortest path between $s'$ and $t'$. On the other hand, since $d\in V_{SC_H[a,b]}=V_{p[a,b]}$ and $u$ appears no later than $a$ on $p$ in the orientation from $s$ to $t$, we have $SC_H[a,d]\subseteq p[u,d]$ and so $|C_{p'}|\ge |E_{SC_H[a,d]}|+|E_{SC_H[d,c]}|$. Since $b\in V_{SC_H[d,c]}$ and $b\neq c$ by \Cref{placee1}, we get $|C_{p'}|>|E_{SC_H[a,b]}|=|C_{p}|$, which contradicts the definition of $p$ whose cycle-overlap should be the longest.
\end{proof}
\begin{claim}\label{doublenon}
$V_{p_1[c,t']}$ doesn't overlap with $V_{p[s,a]}\cup V_{p[b,t]}$. 
\end{claim}
\begin{proof}
The proof is split into two parts:
\begin{itemize}
\item[(1)] $V_{p_1[c,t']}$ doesn't overlap with $V_{p[s,a]}$: Suppose by contradiction that there exists a vertex $u\in V_{p[s,a]}\cap V_{p_1[c,t']}$, then $p_1[b,u]\in P_{u,b}$ should be a shortest path between $u$ and $b$. It follows that $p_1[b,u]=p_1[b,c]\cup p_1[c,u]$, $p_1[c,u]$ is not on the cycle and $p_1[b,c]=SC_H[b,c]\subseteq SC_H[b,a]$. Therefore $e\notin p_1[b,u]$. However in this case $p'=p[s,u]\cup p_1[u,b]\cup p[b,t]\in P_{s,t}$ will be a shortest path between $s$ and $t$ avoiding $e$, which contradicts \Cref{depen} since $e$ is a danger for $(s,t)$.

\item[(2)] $V_{p_1[c,t']}$ doesn't overlap with $V_{p[b,t]}$: We suppose by contradiction that there exists a vertex $u\in V_{p[b,t]}\cap V_{p_1[c,t']}$, then $p_1[b,u]=SC_H[b,c]\cup p_1[c,u]\in P_{b,u}$ should be a shortest path between $b$ and $u$. However, this implies $p'=p[s,a]\cup SC_H[a,b]\cup SC_H[b,c]\cup p_1[c,u]\cup p[u,t]\in P_{s,t}$ is a shortest path between $s$ and $t$, and we can derive $|C_{p'}|\ge|E_{SC_H[a,c]}|>|E_{SC_H[a,b]}|=|C_{p}|$ since $c\in V_{SC_H[b',a]}$ by \Cref{placee1}. This contradicts the definition of $p$ whose cycle-segment should be one of the longest among those of $P_{s,t}$.
\end{itemize}
\end{proof}
By property of shortest paths, there is no intersection between $p[s,a]$ and $p[b,t]$, or $p_1[s',d]$ and $p_1[c,t']$. Therefore, from \Cref{innernon} and \Cref{doublenon}, we know that among the four paths $p[s,a],p[b,t],p_1[s',d],p_1[c,t']$ outside the cycle, only $p_1[s',d]$ and $p[b,t]$ may intersect. This is useful in the following proof.
\begin{figure}[t!]
\centering 
\includegraphics[width=0.5\textwidth]{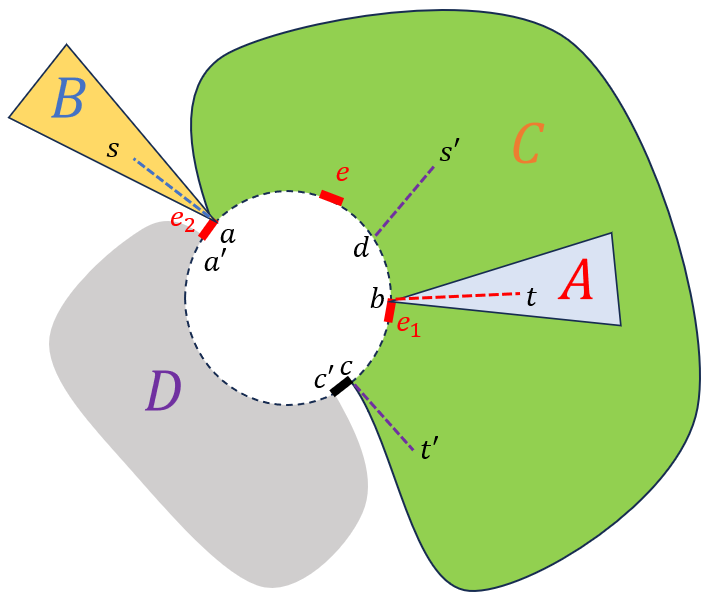} 
\caption{An example of the vertex-partition we use.} 
\label{fig:partition} 
\end{figure}
Now we will try to understand the structure of $H\backslash\{e_2\}$. First, we need to separate vertices in $V_H$ into four categories. These categories are defined as follows. Readers can refer to \Cref{fig:partition} as an example.
\begin{itemize}
\item The set $A\subseteq V_H$ contains $A_0:=V_{p[b,t]}\backslash\{b\}$ and all vertices that are disconnected from the cycle $SC_H$ in $H\backslash A_0$. Let $H_1:=H\backslash A$.

\item The set $B\subseteq V_H$ contains $B_0:=V_{p[s,a]}\backslash\{a\}$ and all vertices that are disconnected from the cycle $SC_H$ in $H_1\backslash B_0$. Let $H_2:=H_1\backslash B$.

\item The set $C\subseteq V_H$ contains $C_0:=V_{SC_H[a,c]}$ and all vertices that are connected with $C_0$ in $H_2\backslash \left(V_{SC_H}\backslash C_0\right)$. Let $H_3:=H_2\backslash C$

\item The set $D:=V_{H_3}\subseteq V_H$ contains all remaining vertices.
\end{itemize}
It's obvious that $H_1,H_2,H_3\subseteq H$ are all connected subgraphs, and $A,B,C,D$ form a partition of $V_H$. For any two vertex-set $U,V\subseteq V_G$ and original edge $ (u,v)\in E_G\cap \left(U\times V\right)$, we want to check whether $\dist_{H\backslash\{e_2\}}(u,v)\leq k$ holds. Given any undirected edge-set $R\subseteq V_G\times V_G$, let $\dist_{H\backslash\{e_2\}}(U,V,R):=\max_{u\in U, v\in V, (u,v)\in R}\left(\dist_{H\backslash\{e_2\}}(u,v)\right)$, we know $H\backslash\{e_2\}$ is a $k$-spanner if and only if
\begin{equation}\label{condition}
\dist_{H\backslash\{e_2\}}(V_G,V_G,E_G)=\max_{U,V\in\{A,B,C,D\}}\dist_{H\backslash\{e_2\}}(U,V,E_G)\leq k
\end{equation}
Next, we will upper bound these terms in groups by case analysis.
\begin{claim}\label{aabb}
$\dist_{H\backslash\{e_2\}}(A,A,E_G)\leq k $ and $\dist_{H\backslash\{e_2\}}(B,B,E_G)\leq k$
\end{claim}
\begin{proof}
We just give a proof for $\dist_{H\backslash\{e_2\}}(A,A,E_G)\leq k$. Another part follows from a similar approach.

Given any two vertices $(u,v)\in \left(A\times A\right)\cap E_G$, we aim to show $\dist_{H\backslash\{e_2\}}(u,v)\leq k$. By definition of $A$, we know $u,v$ is disconnected from the cycle after removing $V_{p[b,t]}$. Therefore, for any shortest path $p'\in P_{u,v}$, there are two possibilities:
\begin{itemize}
\item[(1)] $p'$ doesn't touch any vertex in $V_{p[b,t]}$. In this case, $p'$ doesn't touch any cycle-vertex in $V_{SC_H}$ either. This implies $p'\subseteq H\backslash\{e_2\}$. Since $|E_{p'}|=\dist_H(u,v)\leq k$, we conclude that $\dist_{H\backslash\{e_2\}}(u,v)\leq |E_{p'}|\leq k$.

\item[(2)] $p'$ touches some vertex in $V_{p[b,t]}$. Let $u',v'\in V_{p[b,t]}$ denote the first and last vertices in $V_{p[b,t]}$ appeared on $p'$ in the orientation from $u$ to $v$, we know that there are no cycle-edges on $p'[u,u']$ or $p'[v',v]$. Without loss of generality let's assume $u'$ appears early than $v'$ on $p[b,t]$ in the orientation from $b$ to $t$. Since $p$ is a shortest path and $u',v'\in V_{p[b,t]}$, we know $p[u',v']\in P_{u',v'}$ is certainly a shortest path between $u'$ and $v'$. Therefore, $p'':=p'[u,u']\cup p[u',v']\cup p'[v',v]\in P_{u,v}$ is also a shortest path between $u$ and $v$. Since $p''$ doesn't contain any cycle-edge in $E_{SC_H}$, $\dist_{H\backslash\{e_2\}}(u,v)\leq |E_{p''}|\leq \dist_H(u,v)\leq k$, which is exactly what we want.
\end{itemize}
\end{proof}
\begin{figure}[t!]
\centering 
\includegraphics[width=0.5\textwidth]{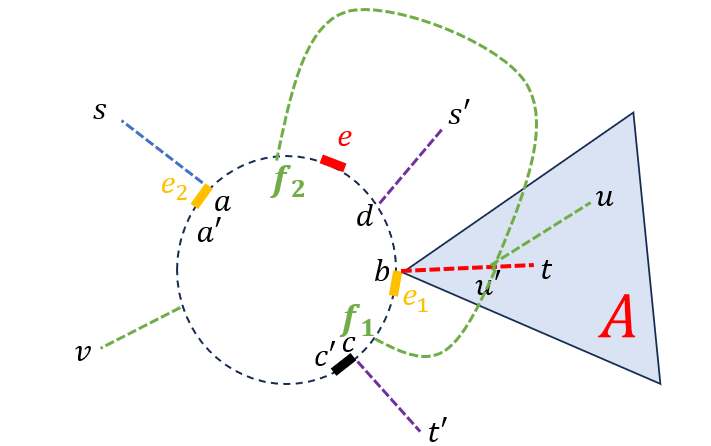} 
\caption{Two possible positions of $f$ in the proof of \Cref{au}. $f_1$ and $f_2$ denote the first and second case respectively.} 
\label{twopos} 
\end{figure}
\begin{claim}\label{au}
$\dist_{H\backslash\{e_2\}}(A,V_G\backslash A,E_G)\leq k$
\end{claim}
\begin{proof}
Suppose by contradiction we have $u\in A$ and $v\in V_G\backslash A$ such that $e_2$ is a danger for $(u,v)\in E_G$. Let $p'=\cd(u,v,e_2)\in P_{u,v}$, from \Cref{depen}, we know $e_2=(a,a')\in E_{p'}$. Moreover, from definition of $A$ we know there exists some vertex $u'\in \left(V_{p[b,t]}\cap V_{p'}\right)$ and $u'$ appears earlier than $a$ on $p'$ in the orientation from $u$ to $v$. There are two possibilities:
\begin{itemize}
\item[(1)] If $a'$ appears later than $a$ in $p'$ oriented from $u$ to $v$, it means $e_2\notin E_{p'[u',a]}$. However, since $p[a,u']=SC_H[a,b]\cup p[b,u']\in P_{a,u'}$ is a shortest path between $u'$ and $a$, we know $p''=p'[u,u']\cup p[u',a]\cup\{e_2\}\cup p'[a',v]\in P_{u,v}$ is a shortest path betwen $u$ and $v$. We can observe that $E_{SC_H[a',b]}\subseteq C_{p''}$ so $|C_{p''}|>|C_{p}|=|E_{SC_H[a,b]}|$, which contradicts the definition of $p$ whose cycle-segment should be the longest.

\item[(2)] If $a'$ appears earlier than $a$ in $p'$ oriented from $u$ to $v$. By \Cref{conse}, we know $e_2\in E_{SC_H[f,a]}\subseteq C_{p'}$ where $f\in V_{SC_H}$ denotes the first cycle-vertex in $p'$ oriented from $u$ to $v$. There are two possible positions for $f$: (See \Cref{twopos})
\begin{itemize}
\item[(a)] $e\notin E_{SC_H[f,a]}$. We know $p'[u',a]=p'[u',f]\cup SC_H[f,a]\in P_{u',a}$ is a shortest path between $u'$ and $a$. Therefore, $p''=p[s,a]\cup p'[a,u']\cup p[u',t]\in P_{s,t}$ is a shortest path between $s$ and $t$. However, since $C_{p''}=C_{p'[a,u']}\subseteq E_{SC_H[f,a]}$, it follows that $e\notin E_{p''}$. This contradicts \Cref{depen} since $e$ is a danger for $(s,t)$. 

\item[(b)] $e\in E_{SC_H[f,a]}$. In this case, $e\notin E_{SC_H[a,f]}$ and $f\in V_{SC_H[a,b]}$. By \Cref{conse} and definition of $f$, $p'[f,u']$ doesn't contain any cycle-edge. Since $f,u'\in V_p$ and $p'[f,u']\in P_{f,u'}$ is a shortest path between $f$ and $u'$, $p''=p[s,a]\cup SC_H[a,f]\cup p''[f,u']\cup p[u',t]\in P_{s,t}$ is a shortest path between $s$ and $t$. However, $e\notin E_{p''}$, which contradicts \Cref{depen}.
\end{itemize}
\end{itemize}
From the discussions above, all possible cases derive a contradiction, so the original assumption mustn't hold. We are done.
\end{proof}
Then, before we dive into other cases, the following simple observation is required
\begin{fact}\label{routelb}
$|A\cup B|\ge k+|E_{SC_H[a,b]}|-L_H$
\end{fact}
\begin{proof}
From \Cref{upg} we can see $|A\cup B|\ge |E_{p[s,a]}|+|E_{p[b,t]}|$. Since $e$ is a danger for $(s,t)$. By \Cref{dangerdef} and previous definitions we have $|E_{p}|=|E_{p[s,a]}|+|E_{SC_H[b,a]}|+|E_{p[b,t]}|> k$, which implies $|A\cup B|\ge |E_{p[s,a]}|+|E_{p[b,t]}|> k-|E_{SC_H[b,a]}|=k+|E_{SC_H[a,b]}|-L_H$.
\end{proof}
The next goal is to prove the following case.
\begin{claim}\label{cdin}
$\dist_{H\backslash\{e_2\}}(C\cup D,C\cup D,E_{G})\leq k$
\end{claim}
\begin{proof}
Since $H$ is connected and $e_2$ is a cycle-edge, $H\backslash \{e_2\}$ must be connected. Also, by the definitions of $C$ and $D$, we know the vertex-set $C\cup D$ contains all cycle-vertices $V_{SC_H}$, and all vertices of $C\cup D$ are from the same connected component after removing $A\cup B$ from $H$. This implies $C\cup D$ is connected on the subgraph $H_2'=H_2\backslash\{e_2\}\subseteq H\backslash\{e_2\}$. By \Cref{routelb}, we have
\begin{equation}\label{eq:allcd}
|V_{H_2'}|=|V_{H}|-|A\cup B|\leq n-k+L_H<3(n-k)\leq k
\end{equation}
The last inequality is from $k>\frac{3}{4}n+O(1)$. From \Cref{eq:allcd} and connectedness of $C\cup D$ in $H'_2$ discussed before, for any $u,v\in (C\cup D)\times (C\cup D)$, we know there is a path from $u$ to $v$ in $H_2'$ with length at most $|V_{H_2}|\leq k$, which implies $\dist_{H\backslash\{e_2\}}(u,v)\leq \dist_{H_2'}(u,v)\leq k$. We are done.
\end{proof}
The only remaining case is to show $\dist_{H\backslash \{e_2\}
}(B,C\cup D,E_G)\leq k$. The following simple observation is required.
\begin{fact}\label{cyclepartup}
$|E_{SC_H[b,c]}|\leq |E_{SC_H[a,b]}|$
\end{fact}
\begin{proof}
We know $E_{SC_H[b,c]}\subseteq C_{p_1}$ where $p_1\in P_{s',t'}$ and $(s',t')$ is endangered by $e_1$. We also know $E_{SC_H[a,b]}=C_{p}$. Therefore, by definition of $p$ whose cycle-segment should be the largest among all shortest paths between endangered pairs, we must have $|C_p|\ge |C_{p_1}|$ so the proof is done.
\end{proof}
\begin{claim}\label{bclb}
$\dist_{H\backslash\{e_2\}}(B,C,E_G)\leq k$
\end{claim}
\begin{proof}
By definitions of $B,C$, vertices in $B\cup C$ are connected in $H_1$. Moreover, for any original edge $(u,v)\in (B\times C)\cap E_G$, by simple observation there is a path $p_a$ between $u$ and $a$ that only touches vertices in $B\cup \{a\}$ and another path $p_b$ between $a$ and $v$ only using vertices in $C$. Therefore, there is a path $p'=p_a\cup p_b$ between $u$ and $v$ that only touches vertices in $B\cup C$. Since $a'\notin B\cup C$, we know $e_2\notin E_{p'}$, $p'\subseteq H\backslash \{e_2\}$ and $|E_{p'}|<|B\cup C|=|V_G|-|A|-|D|$. By definition since we've assumed $|E_{p[s,a]}|\leq |E_{p[b,t]}|$, by \Cref{routelb} we have 
\begin{equation*}
|A|\ge |E_{p[b,t]}|\ge \frac{1}{2}\left(|E_{p[s,a]}|+|E_{p[b,t]}|\right)\ge \frac{1}{2}\left(k+|E_{SC_H[a,b]}|-L_H\right).
\end{equation*}
On the other hand, we can also derive a simple lower bound on $|D|$: Since $D\cap V_{SC_H}=V_{SC_H}-V_{SC_H[a,c]}$, we have
\begin{equation*}
|D|\ge |V_{SC_H}|-|V_{SC_H[a,c]}|=L_H-|E_{SC_H[a,b]}|-|E_{SC_H[b,c]}|-1\ge L_H-2|E_{SC_H[a,b]}|-1.
\end{equation*}
The last inequality is from \Cref{cyclepartup}. Finally, we get
\begin{align*}
|E_{p'}|&<|V_G|-|A|-|D|\\
&\leq n-\frac{1}{2}\left(k+|E_{SC_H[a,b]}|-L_H\right)-\left(L_H-2|E_{SC_H[a,b]}|-1\right)\\
&= n-\frac{k}{2}-\frac{L_H}{2}+\frac{3|E_{SC_H[a,b]}|}{2}+1\\
&\leq n-\frac{k}{2}+\frac{L_H}{4}+1\\
&\leq \frac{3}{2}n-k+1
\end{align*}
The last two inequalities above are from the following two facts: 
\begin{itemize}
\item[(1)]$SC_H[a,b]=p[a,b]\in P_{a,b}$ is a shortest path between $a$ and $b$, so there must be $|E_{SC_H[a,b]}|\leq |E_{SC_H[b,a]}|$ and $|E_{SC_H[a,b]}|\leq \frac{L_H}{2}$.
\item[(2)]$L_H<2(n-k)$
\end{itemize}
When $k>\frac{3}{4}n+O(1)$, we have $|E_{p'}|\leq \frac{3}{2}n-k\leq k$. Since $p'\in H\backslash\{e_2\}$, we conclude that for any edge $(u,v)\in (B\times C)\cap E_G$, $\dist_{H\backslash\{e_2\}}(u,v)\leq k$, which proves our claim.
\end{proof}
The last task is to show 
$\dist_{H\backslash\{e_2\}}(B,D,E_G)\leq k$.
\begin{claim}\label{bdlb}
$\dist_{H\backslash\{e_2\}}(B,D,E_G)\leq k$
\end{claim}
\begin{proof}
For any original edge $(u,v)\in (B\times D)\cap E_G$, we know $\dist_H(u,v)\leq k$. We will prove $\dist_{H\backslash\{e_2\}}(u,v)\leq k$. 
\begin{figure}[t!]
\centering 
\includegraphics[width=1.0\textwidth]{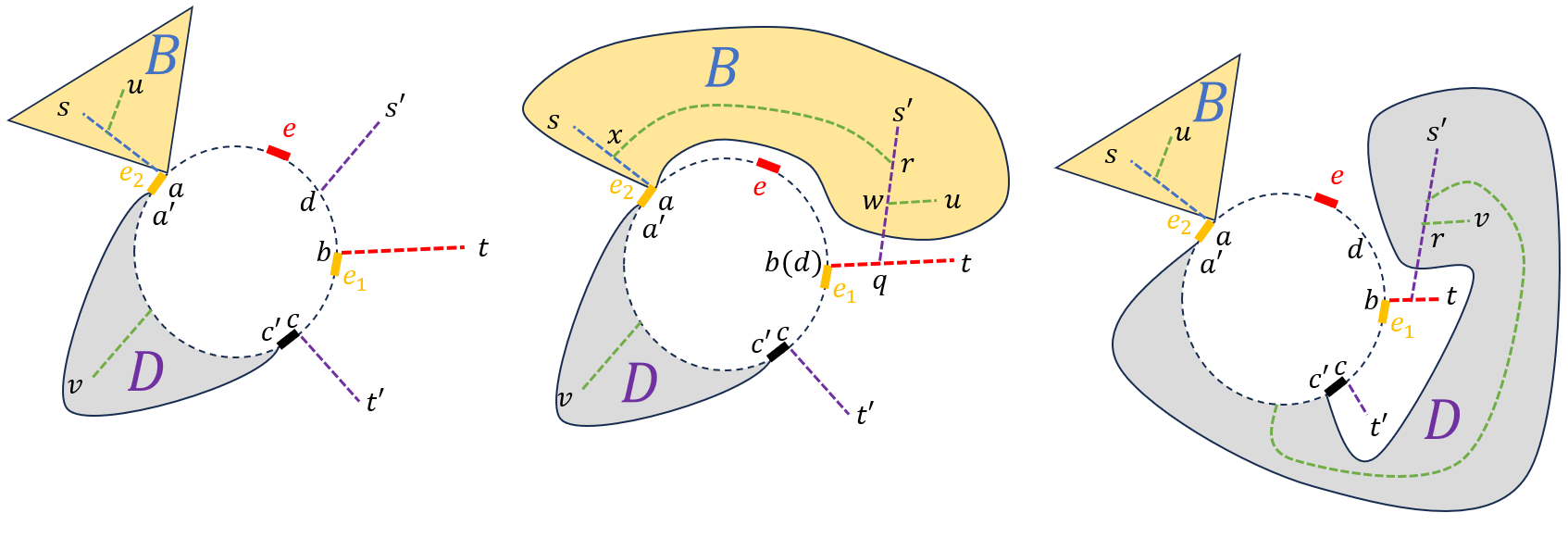} 
\caption{The left/middle/right graph illustrates the Case (1)/Case (2)/Case (3) in the proof of \Cref{bdlb}} 
\label{fig:excases} 
\end{figure}
By \Cref{doublenon}, $E_{p_1[c,t']}$ doesn't overlap with any edge in $E_p$, and it is connected to the cycle-vertex $c$. Therefore, by definition of the four vertex-classes, $V_{p_1[c,t']}\subseteq C$. Since $H[B]$ is connected, there must be a path $p_a$ between $u$ and $a$ in $H[B]$. On the other hand, since $H[D]$ is connected, there must be a path $p_b$ between $v$ and $c'\in V_{SC_H}$ in $H[D]$ where $(c,c')\in E_{SC_H}$ is a cycle-edge. Consider the path $p'=p_a\cup SC_H[a,c]\cup\{(c,c')\}\cup p_b$ between $u$ and $v$ in $H[B]\cup H[D]\cup SC_H[a,a']\subseteq H\backslash\{e_2\}$, then there are three possibilities. \Cref{fig:excases} gives drawing examples for these cases.
\begin{itemize}
\item[(1)] When neither $p_a$ nor $p_b$ touches vertices in $V_{p_1(d,s']}$. In this case, from the above discussion we've known $p'$ doesn't touch any vertex in $V_{p_1(d,s']}\cup V_{p_1(c,t']}$. Therefore, from a similar argument on $V_{p_1(d,s']}\cup V_{p_1(c,t']}$ as in \Cref{routelb}, we can get an upper bound on $|E_{p'}|$
\begin{equation*}
|E_{p'}|\leq |V_G|-|E_{p_1[d,s']}|-|E_{p_1[c,t']}|\leq n-(k+E_{SC_H}[a,b]-L_H)\leq 3(n-k)
\end{equation*}
When $k>\frac{3}{4}n+O(1)$, we derive $\dist_{H\backslash\{e_2\}}\leq |E_{p'}|\leq k$.
\item[(2)] When $p_a$ touches vertices in $V_{p_1(d,s']}$. In this case, there must be $b=d$ and $V_{p_1(d,s']}\in A\cup B$. As in \Cref{fig:excases}, there must be a `separation' vertex $q$ distinguishing $p[b,t]$ and $p_1[d,s']$. Without loss of generality, we can choose $p_a$ such that edges in $E_{p_a}\cap E_{p(q,s']}$ are continuous in $p_a$. Using this fact, we can assume there is a vertex $r\in V_{p_1[s',q]}\cap V_{p_a}$ such that $V_{p_a[a,r)}$ doesn't touch any vertex in $V_{p_1[s',q]}$. Let $w\in V_{p_a[u,r]}\cap V_{p_1[r,q]}$ denote the vertex that  $V_{p_a(w,u]}\cap V_{p_a[r,q]}=\emptyset$ holds. Similarly, we can assume edges in $E_{p_a}\cap E_{p[s,a]}$ are continuous in $p_a$ and let $x\in V_{p_a[r,a]}\cap V_{p[s,a]}$ denote the vertex such that $V_{p[s,a]}\cap V_{p_a(x,u]}=\emptyset$ and $p_a[x,a]=p[x,a]$. We can upper bound the size $|E_{p'}|$ as
\begin{align}
|E_{p'}|&\leq |V_G|-1-|A|-|V_{p_1(w,q)}|-|V_{p(x,s]}|\\ \label{firstpathsize}
&=n-(|E_{p_1[q,w]}|+|E_{p[b,q]}|+|E_{p[q,t]}|+|E_{p[x,s]}|)
\end{align}
On the other hand, we consider another path $p''=p_a[u,w]\cup p_1[w,b]\cup SC_H[b,c']\cup p_b\subseteq H\backslash\{e_2\}$ between $u$ and $v$. In this case, it doesn't touch any vertices in $V_{p_a(x,w)}\cup V_{p(a,s]}\cup V_{p(q,t]}$. These vertex-sets are pairwise disjoint by \Cref{innernon}, \Cref{doublenon}, and other previous discussions. Therefore, we can upper bound $|E_{p''}|$ as
\begin{align}
|E_{p''}|&\leq |V_G|-1-|V_{p_a(x,w)}|-|V_{p(a,s]}|-|V_{p(q,t]}|\\\label{secondpathsize}
&=n-(|E_{p_a[r,x]}|+|E_{p_1[r,w]}|+|E_{p[q,t]}|+|E_{p[a,x]}|+|E_{p[x,s]}|)
\end{align}
Consider the path $p_0=p[q,t]\cup p_1[q,w]\cup p_1[w,r]\cup p_a[r,x]\cup p[x,s]$ between $s$ and $t$. Since $e\notin E_{p_0}$ and $(s,t)$ is endangered by $e$, by \Cref{depen} we know
\begin{equation}\label{interpart}
|E_{p_0}|=|E_{p[q,t]}|+|E_{p_1[q,w]}|+|E_{p_1[w,r]}|+|E_{p_a[r,x]}|+|E_{p[x,s]}|>k
\end{equation}
From \Cref{interpart} and \Cref{routelb}, it follows that
\begin{align}
&\max\left(|E_{p_1[q,w]}|+|E_{p[b,q]}|,|E_{p_a[r,x]}|+|E_{p_1[r,w]}|+|E_{p[a,x]}|\right)+|E_{p[q,t]}|+|E_{p[x,s]}|\\
\ge&\frac{1}{2}\left(|E_{p[q,t]}|+|E_{p[b,q]}|+|E_{p[a,x]}|+|E_{p[x,s]}|+|E_{p[x,s]}|+|E_{p_a[r,x]}|+ |E_{p_1[r,q]}|+|E_{p[q,t]}|\right)\\
\ge&\frac{1}{2}\left(|E_{p[a,s]}|+|E_{p[b,t]}|+|E_{p_0}|\right)\\
\label{tworoutelb}\ge&\frac{2k-L_H}{2}\ge 2k-n
\end{align}
We combine \Cref{firstpathsize,secondpathsize,tworoutelb}, it gives an upper bound on $\min\left(|E_{p'}|,|E_{p''}|\right)$
\begin{equation*}
\min\left(|E_{p'}|,|E_{p''}|\right)\leq n-(2k-n)=2n-2k
\end{equation*}
When $k>\frac{3}{4}n+O(1)$, the above equation gives us $\dist_{H\backslash\{e_2\}}(u,v)\leq \min\left(|E_{p'}|,|E_{p''}|\right)\leq k$.
\item[(3)] When $p_b$ touches vertices in $V_{p_1(d,s']}$. In this case, we can observe $V_{p_1(d,s']}\in A\cup D$ and so $p_a$ doesn't touch $p_1(d,s']$. Without loss of generality, we can assume there exists a vertex $r\in V_{p_1(d,s']}\cap p_b$ such that $p_b[r,v]$ doesn't contain any cycle-vertex in $V_{SC_H}$. Let $p''=p_a\cup SC_H[a,d]\cup p_1[d,r]\cup p_b[r,v]\subseteq H\backslash\{e_2\}$ be a path between $u$ and $v$. Since $p''$ doesn't touch any vertex in $SC_H[b',a']$ we know $|E_{p''}|\leq n-|E_{SC_H[b,a]}|=n+|E_{SC_H[a,b]}|-L_H$. When $L_H-|E_{SC_H[a,b]}|\ge n-k$, we've got $\dist_{H\backslash\{e_2\}}(u,v)\leq |E_{p''}|\leq k$ and we are done. The only remaining task is to argue the case when $L_H-|E_{SC_H[a,b]}|<n-k$. In this case, since $p'$ doesn't contain any vertex in $A$. By \Cref{routelb} and assumption that $|E_{p[s,a]}|\leq |E_{p[b,t]}|$, we derive $|E_{p'}|\leq n-|A|\leq n-\frac{k+|E_{SC_H[a,b]}|-L_H}{2}\leq \frac{3}{2}n-k\leq k$ when $k>\frac{3}{4}n$. Therefore we can conclude $\dist_{H\backslash\{e_2\}}(u,v)\leq |E_{p'}|\leq k$ in this last case.
\end{itemize}
\end{proof}
Combine all results \Cref{aabb}, \Cref{au}, \Cref{cdin}, \Cref{bclb} and \Cref{bdlb}, we can prove \Cref{condition}, and our lemma holds
\end{proof}.

Now we are ready to prove \Cref{exgoodub}.

\begin{theorem}[Restated, \Cref{exgoodub}]
There is a deterministic polynomial time algorithm $A$ such that for all sufficiently large $n$ and any $k>\frac{3}{4}n+O(1)$, given any $n$-vertex graph $G$ and its $k$-spanner $H\subseteq G$, $A(G,H,k)$ outputs a subgraph $R\subseteq H$ of $H$ such that $R$ is a $k$-spanner of $G$ and $R$ has girth at least $k+2$.
\end{theorem}
\begin{proof}
Given any spanner $H$ of $G$, we first find one of its smallest cycles and see whether $L_H\ge k+2$. If so, we are done. Otherwise, if $2(n-k)\leq L_H\leq k+1$, we apply \Cref{largecycle} to break the smallest cycle while preserving the transformed subgraph is still a $k$-spanner. By repeatedly running the above procedure, the only remaining case is $L_H<2(n-k)$. In this case, we apply \Cref{exsmallcycle} to break the smallest cycle. By repeatedly doing the above procedure, we can finally find a subgraph $R\subseteq H$ with girth at least $k+2$ and is still a $k$-spanner of $G$.
\end{proof}
\subsection{Upper Bound for Good Pairs}\label{sec:goodpair}
The next part is to prove \Cref{goodub}. It's kind of similar to the proof of \Cref{exsmallcycle}. In the algorithm given in \Cref{exgoodub}, each time when we find $L_H<2(n-k)$, we are guaranteed to find a cycle-edge, remove it, and the remaining subgraph is still a $k$-spanner. However, in \Cref{goodub}, each time we will remove a cycle-edge from the smallest cycle, and then (possibly) add an original edge in $E_G$ to $H$ in order to maintain the requirement of $k$-spanner. Our newly added edge won't create any new cycle shorter than $k+2$, so each time the deleted cycle-edge must be from a cycle with length at most $k+1$ in the original $k$-spanner. which means an edge won't be deleted twice among iterations. Therefore, the algorithm can terminate within polynomial time and will finally give us a $k$-spanner with girth at least $k+2$ and its size won't become larger.

\begin{figure}[t!]
\centering 
\includegraphics[width=0.4\textwidth]{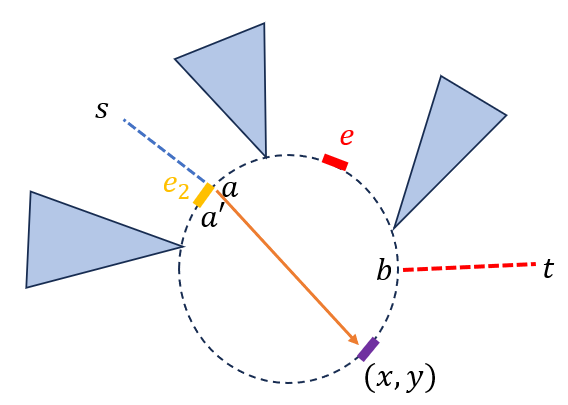} 
\caption{This graph illustrates the $a$-originated shortest path tree $T$. Note that all above edges are in $T$ except for the purple edge $(x,y)$} 
\label{tree} 
\end{figure}

We follow from the framework built in \Cref{sec:common}. Let $E_0$ denote the set of edges in $p[s,a]$, $p[a,t]$ and $SC_H$ 
except for the antipode cycle-edge of $a$ on $SC_H$ (The one cycle-edge $(x,y)$ such that exactly one of $|E_{SC_H[x,a]}|,|E_{SC_H[y,a]}|$ is larger than $\frac{L_H}{2}$), we construct a shortest path tree $T$ of $H$ from origin $a$ to every other vertices. In fact, we can construct such a shortest path tree $T$ satisfying $E_0\subseteq E_T$. Please see \Cref{tree} for an example. For any two vertices $u,w$, let $T[u,w]$ denote the tree-path from $u$ to $w$ on $T$. The following lemma is the crucial building block of \Cref{goodub}.

\begin{lemma}\label{gdsmallcycle}
If $k>\frac{2}{3}n+O(1)$, and $H\backslash\{e_2\}$ is not a $k$-spanner of $G$ where $H$ is defined as in  \Cref{sec:common} with $L_H<2(n-k)$. We can find at most one original edge $(u,v)\in E_G$ such that $H':=\left(H\backslash\{e_2\}\right)\cup \{(u,v)\}$ is still a $k$-spanner, and $(u,v)$ is not in any cycle with length at most $k+1$ in $H'$.
\end{lemma}
\begin{proof}
Since $H\backslash\{e_2\}$ isn't a $k$-spanner of $G$, there must be a bunch of original edges $(u_1,v_1),\dots,(u_m,v_m)\in E_G$ that are endangered by $e_2$. By \Cref{depen}, for any $(u_i,v_i)$, $e_2$ must be in exactly one of $T[u_i,a]$ and $T[a,v_i]$. Without loss generality, let's say $e_2\in T[u_i,a]$ and $e_2\notin T[a,v_i]$ for any pair $(u_i,v_i)$. We use $u'_i,v'_i\in V_{SC_H}$ to denote the first vertices on cycle $SC_H$ appeared on $T[u,a]$ and $T[v,a]$ in orientations from $u$ to $a$ and $v$ to $a$ respectively. By a similar argument from \Cref{cyclepartup}, we know $|E_{SC_H[u_i',v_i']}|\leq |E_{SC_H[a,b]}|$. Also, as in \Cref{placee1}, we can assume $u_i'\in V_{SC_H[b,a']}$ and $v_i'\in V_{SC_H[a,b]}$.

Let $(u_g,v_g)$ denote the one $e_2$-endangered pair such that $v'_g$ is closest to $a$ among $g\in[m]$, and if there are ties, $u'_g$ should be as close to $a'$ as possible. we just add $(u_g,v_g)\in E_G$ to $H$. Since $(u_g,v_g)$ is endangered by $e_2$,  we have $\dist_{H\backslash\{e_2\}}(u_g,v_g)>k$ so $(u_g,v_g)$ won't be in any cycle with length at most $k+1$ in $H'=\left(H\backslash\{e_2\}\right)\cup\{(u_g,v_g)\}$. The only remaining part is to show that for any other endangered demand pairs $(u_i,v_i)$, $\dist_{H'}(u_i,v_i)\leq k$. Fix any $(u_i,v_i),i\in[m]$, we will do case analysis on two possibilities:

\begin{figure}[t!]
\centering 
\includegraphics[width=0.8\textwidth]{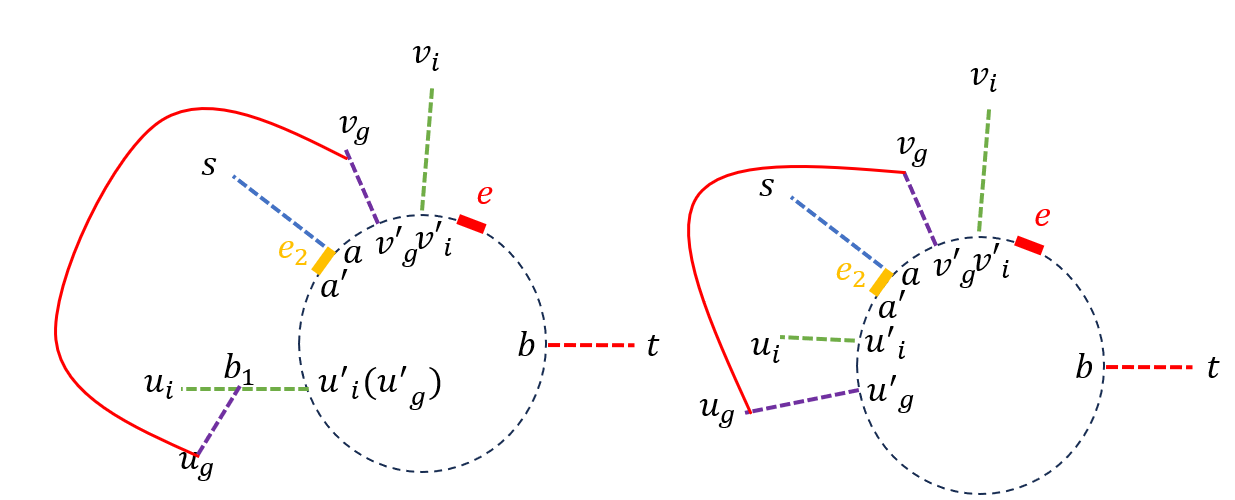} 
\caption{Two cases of $(u_i,v_i)$ in \Cref{gdsmallcycle}} 
\label{cases} 
\end{figure}
\begin{itemize}
\item[(1)] If $u'_i=u_g'$ or $v'_i=v_g'$: Without loss of generality we assume $u'_i=u'_g$, and the other case follows similarly. By definition of $u_g,v_g$ we have $u_i'\in V_{SC_H[b,u_g']}$ and $v_i'\in V_{SC_H[v_g',b]}$.
This case is shown in the first picture of \Cref{cases}. By a similar argument with \Cref{innernon} and \Cref{doublenon}, we know neither $V_{T[u_i,u_i']}$ nor $V_{T[v_i
,v_i]}$ overlaps with $V_{p[b,t]}$. Therefore, consider the path $p'=T[u_i,u_g']\cup T[u_g',u_g]\cup\{(u_g,v_g)\}\cup T[v_g,v'_g]\cup T[v_g',v_i']\cup T[v_i',v_i]$ is a path between $u_i$ and $v_i$ that avoids all vertices in $V_{p[b,t]}\cup V_{SC_H(v_i',u_i')}\cup V_{SC_H(u_g',v_g')}$. This gives us the following:
\begin{align*}
|E_{p'}|&< |V_{G}|-1-|V_{p(b,t]}|-|V_{SC_H(u_g',v_g')}|-|V_{SC_H(v_i',u_i')}|\\
&\leq n-\frac{k+|E_{SC_H[a,b]}|-L_H}{2}-|E_{SC_H[u_g',v_g']}|-(L_H-|E_{SC_H[u_i',v_i']}|)+2\\
&\leq n-\frac{k}{2}-\frac{L_H-|E_{SC_H[a,b]}|}{2}+2\\
&\leq n-\frac{k}{2}+2
\end{align*}
We use the fact $|E_{SC_H[u_i',v_i']}|\leq |E_{SC_H[a,b]}|$ discussed earlier in the above evaluation. When $k>\frac{2}{3}n+O(1)$, it's easy to see that $|E_{p'}|\leq k$. Since $p'$ doesn't touch $e_2\in E_{SC_H[u_g',v_g']}$, we get $\dist_{H'}(u_i,v_i)\leq |E_{p'}|\leq k$.

\item[(2)] If $u_i'\neq u_g'$ and $v_i'\neq v_g'$: By definition of $(u_g,v_g)$, there must be $v_i'\in V_{SC_H[v_g',b]}$. If $u_i'\in V_{SC_H[b,u_g']}$, we can just use the same argument as the first case to conclude $\dist_{H'}(u_i,v_i)\leq k$. Therefore, the only remaining case is when $u_i'\in V_{SC_H[u_g',a']}$ as the second picture of \Cref{cases} shows. In this case, consider the path $p'=T[u_i,u_i']\cup SC_H[u_g',u_i']\cup T[u_g',u_g]\cup\{(u_g,v_g)\}\cup T[v_g,v_g']\cup T[v_g',v_i]$. We can see $p'$ doesn't touch any vertex in $V_{p(b,t]}$, $V_{SC_H(u_i',v_g')}$ or $V_{SC_H(v_i',u_g')}$. We can bound the size of $E_{p'}$ as:
\begin{align*}
|E_{p'}|&<|V_G|-|V_{p(b,t]}|-|V_{SC_H(v_i',u_g')}|-|V_{SC_H(u_i',v_g')}|\\
&\leq n-\frac{k+|E_{SC_H[a,b]}|-L_H}{2}-(L_H-|E_{SC_H[u_g',v_g']}|-|E_{SC_H[u_i',v_i']}|)+3\\
&\leq n-\frac{k}{2}-\frac{L_H}{2}+|E_{SC_H[u_g',v_g']}|+\frac{|E_{SC_H[a,b]}|}{2}+3
\end{align*}
Since $p'$ doesn't contain $e_2$, if the above is equal or smaller than $k$, we get $\dist_{H'}(u_i,v_i)\leq |E_{p'}|\leq k$ and done. Otherwise, it must be the case that
\begin{equation}\label{tradeoff}
|E_{SC_H[u_g',v_g']}|\ge \frac{3}{2}k-n+\frac{L_H-|E_{SC_H[a,b]}|}{2}
\end{equation}
Now, consider the path $p''=T[u_i,u_i']\cup SC_H[v_i',u_i']\cup T[v_i',v_i]$ between $u_i$ and $v_i$. Since $u_i',v_i',u_g',v_g'$ are four distinct cycle-vertices, we can observe that neither $T[u_i,u_i']$ nor $T[v_i',v_i]$ has vertex on $T[v_g,v_g')$ or $T[u_g,u_g']$. we know that $p''$ doesn't touch any vertex on $p(b,t]$, $T[u_g,u_g')$ or $T[v_g,v_g')$. By a similar argument to \Cref{routelb}, we have $|V_{T[u_g,u_g')}|+|V_{T[v_g,v_g')}|>k+|E_{SC_H[u_g',v_g']}|-L_H$. Therefore, we can bound the size of $p''$ as:
\begin{align*}
|E_{p''}|&<|V_G|-|V_{p(b,t]}|-|V_{T[u_g,u_g')}|-|V_{T[v_g,v_g')}|\\
&\leq n-\frac{k+|E_{SC_H[a,b]}|-L_H}{2}-(k+|E_{SC_H[u_g',v_g']}|-L_H)\\
&\leq n-\frac{3}{2}k+\frac{L_H-|E_{SC_H[a,b]}|}{2}-|E_{SC_H[u'_g,v'_g]}|+L_H\\
&\leq 2n-3k+L_H\\
&\leq 4n-5k
\end{align*}
The above calculation uses \Cref{tradeoff} and the fact $L_H\leq 2(n-k)$. When $k>\frac{2}{3}n+O(1)$, we have $|E_{p''}|\leq k$. Since $p''$ also doesn't touch $e_2$, we get $\dist_{H'}(u_i,v_i)\leq|E_{p''}|\leq k$.
\end{itemize}
Since for any $(u_i,v_i),i\in[m]$ endagnered by $e_2$, we've shown $\dist_{H'}(u_i,v_i)\leq k$, $H'=\left(H\backslash\{e_2\}\right)\cup\{(u_g,v_g)\}$ must be a $k$-spanner.
\end{proof}
Now we are ready to prove \Cref{goodub}.
\begin{theorem}[Restated, \Cref{goodub}]\label{goodubit}
There is a deterministic polynomial time algorithm $A$ such that for all sufficiently large $n$ and any $k>\frac{2}{3}n+O(1)$, given any $n$-vertex graph $G$ and its $k$-spanner $H=(V,E_H)\subseteq G$, $A(G,H,k)$ outputs $R=(V,E_R)\subseteq G$ such that $R$ is a $k$-spanner of $G$ with girth at least $k+2$. Moreover, the size of 
 $R$ satisfies $|E_R|\leq|E_H|$.   
\end{theorem}
\begin{proof}
Given any spanner $H$, we first find its smallest cycle and see whether $L_H>k+1$. If so, we are done. Otherwise, if $2(n-k)\leq L_H\leq k+1$, we apply \Cref{largecycle} to break the smallest cycle while preserving the property of $k$-spanner. If $L_H<2(n-k)$, we apply \Cref{gdsmallcycle} to break the smallest cycle while preserving the modified graph is still a $k$-spanner, without size blow-up, and no newly added cycle has length smaller than $k+2$. By repeatedly doing the above procedure, we can finally find a subgraph $R\subseteq G$ with girth at least $k+2$ and is still a $k$-spanner of $G$. Moreover, $|E_R|\leq |E_H|$.
\end{proof}
\section{Proofs for Approximately-Good Upper Bounds}\label{sec:approx}
\subsection{Upper Bound for $(2,O(1))$-approx Good}\label{sec:2approx}
In this paragraph, we aim to prove \Cref{2approxub}. This proof is based on the framework built in \Cref{sec:goodpair}, so the notations will follow from \Cref{sec:goodpair}. The following is the central lemma we need.
\begin{lemma}\label{2approxmod}
If $k>\frac{4}{7}n+O(1)$, and $H$ is a $k$-spanner of $G$ with smallest cycle of length at most $k+1$, we can find at most two $e_2$-endangered pairs $(u_g,v_g),(u_h,v_h)\in E_G$ such that $H'=\left(H\backslash\{e_2\}\right)\cup\{(u_g,v_g),(u_h,v_h)\}$ is still a $k$-spanner, and the newly added edges in $H'$ are not in any cycle with length at most $k+1$.
\end{lemma}
\begin{proof}
If $H\backslash\{e_2\}$ is a $k$-spanner, we are done. Suppose not, there are a bunch of original edges $(u_1,v_1),\dots,(u_m,v_m)\in E_G$ endangered by $e_2$ as discussed in \Cref{gdsmallcycle}. Let $(u_g,v_g)$ be the endangered pair such that $|E_{SC_H[u_g',v_g']}|$ is minimized (You can break ties arbitrarily). If $H''=\left(H\backslash\{e_2\}\right)\cup(u_g,v_g)$ is a $k$-spanner, we are done. If it's still not, we choose $(u_h,v_h),h\in[m]$ to be the $e_2$-endangered pairs in $H''$ such that $|E_{SC_H[u_h',v_h']}|$ is minimized. 
Let $H'=\left(H\backslash\{e_2\}\right)\cup\{(u_g,v_g),(u_h,v_h)\}$, we now argue that for any  $(u_i,v_i),i\in[m]$, $\dist_{H'}(u_i,v_i)\leq k$. Let $R=\{i:\dist_{H'}(u_i,v_i)>k\}$. Suppose by contradiction $R$ is non-empty, our goal is to derive some contradiction. Next, we give a technical lemma required
\begin{figure}[t!]
\centering 
\includegraphics[width=1.0\textwidth]{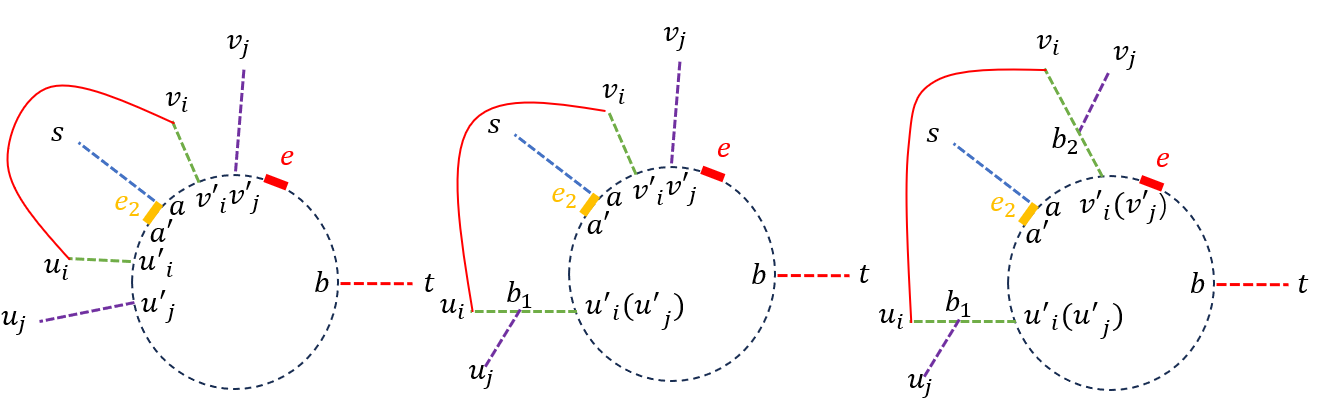} 
\caption{Three cases of intersection patterns of tree paths in the proof of \Cref{routelbplus}} 
\label{casesover} 
\end{figure}
\begin{claim}\label{routelbplus}
For any distinct $i,j\in R\cup \{g,h\}$ such that at least one of them is $g$ or $h$, we define $\mathsf{Ex}(i,j):=|\left(E_{T[u_i,u_i']}\cup E_{T[v_i,v_i']}\right)\backslash\left(E_{T[u_j,u_j']}\cup E_{T[v_j,v_j']}\right)|$ and $\mathsf{Ex}(j,i)$ similarly. Then, both $\mathsf{Ex}(i,j)$ and $\mathsf{Ex}(j,i)$ are at least $\frac{k+1+2|E_{SC_H[u_g',v_g']}|-L_H}{2}$
\end{claim}
\begin{proof}
$T[u_i,u_i'], T[v_i,v_i'], T[u_j,u_j']$ and $T[v_j,v_j']$ are all segments on some shortest path to the root $a$, so they are all tree paths from a vertex to its ancestor. By \Cref{depen}, since $u_i,u_j$ are on the subtree initiated by edge $e_2$ and $v_i,v_j$ are on other subtrees, we know $V_{T[u_i,u_i']}\cup V_{T[u_j,u_j']}$ is disjoint from $V_{T[v_i,v_i']}\cup V_{T[v_j,v_j']}$. There are three cases of the intersection pattern of these tree paths. Without loss of generality, we can assume $i\in \{g,h\}$, $|E_{SC_H[u_i',v_i']}|\leq |E_{SC_H[u_j',v_j']}|$ (By the definition of $g,h$) and $\dist_{\left(H\backslash\{e_2\}\right)\cup\{(u_i,v_i)\}}(u_j,v_j)>k$. We have to argue both of $\maex(i,j),\maex(j,i)\ge \frac{k+1+2|E_{SC_H[u_g',v_g']}|-L_H}{2}$. There are three cases about intersection patterns of these tree paths as shown in \Cref{casesover}. Let's analyze them one by one.
\begin{itemize}
\item[(1)] If $V_{T[u_i,u_i']}$ doesn't intersect with $V_{T[u_j,u_j']}$ and neither does $V_{T[v_i,v_i']}$ with $V_{T[v_j,v_j']}$, it's the case as shown in the first picture of \Cref{casesover}. In this case by a similar argument as in \Cref{routelb}, we get $\maex(i,j)=|E_{T[u_i,u_i']}\cup E_{T[v_i,v_i']}|\ge k+1+|E_{SC_H[u_i',v_i']}|-L_H\ge \frac{k+1+2|E_{SC_H[u_g',v_g']}|-L_H}{2}$. The last inequality is from the definition of $g$ that $|E_{SC_H[u_g',v_g']}|$ is minimized. $\maex(j,i)\ge \frac{k+1+2|E_{SC_H[u_g',v_g']}|-L_H}{2}$ follows from the same argument.

\item[(2)] If $V_{T[v_i,v_i']}$ doesn't intersect with $V_{T[v_j,v_j']}$ but $V_{T[u_i,u_i']}$ and $V_{T[u_j,u_j']}$ does (Or vice versa symmetrically), this case is shown in the second graph of \Cref{casesover}. Let $b_1$ be the lowest common ancestor of $u_i$ and $u_j$, we know $b_1$ is on the path $T[u_i,u_i']$ and $T[u_j,u_j']$ and $u_i'=u_j'$. Since $|E_{SC_H[u_i',v_i']}|\leq |E_{SC_H[u_j',v_j']}|$, we must have $v_j'\in V_{SC_H[v_i',b]}$. Let $x:=|E_{SC_H[u_i',v_i']}|,y:=|E_{SC_H[v_i',v_j']}|$, we give a series of inequalities:
\begin{align}
\label{iless}|E_{T[u_i,b_1]}|+|E_{T[b_1,u_i']}|+|E_{T[v_i,v_i']}|&\leq k-x\\
\label{jless}|E_{T[u_j,b_1]}|+|E_{T[b_1,u_i']}|+|E_{T[v_j,v_j']}|&\leq k-x-y\\
\label{imore}|E_{T[u_i,b_1]}|+|E_{T[b_1,u_i']}|+|E_{T[v_i,v_i']}|&\ge k+1+x-L_H\\
\label{jmore}|E_{T[u_j,b_1]}|+|E_{T[b_1,u_i']}|+|E_{T[v_j,v_j']}|&\ge k+1+x+y-L_H\\
\label{combine}|E_{T[u_j,b_1]}|+|E_{T[b_1,u_i]}|+|E_{T[v_i,v_i']}|+|E_{T[v_j',v_j]}|&\ge k-y
\end{align}
\Cref{iless} and \Cref{jless} are from the fact that $(u_i,v_i),(u_j,v_j)$ are endangered pairs of $e_2$, so by \Cref{depen} $T[u_i,v_i]\in P_{u_i,v_i},T[u_j,v_j]\in P_{u_j,v_j}$ are corresponding shortest paths which has length at most $k$ by property of $k$-spanner $H$. \Cref{imore} and \Cref{jmore} are from basically the same argument as \Cref{routelb} when considering $(u_i,v_i),(u_j,v_j)$ as $e_2$-endangered pairs and thus has distance larger than $k$ in $H\backslash\{e_2\}$. \Cref{combine} can be derived as follows: Consider path $p'=T[u_j,b_1]\cup T[b_1,u_i]\cup\{(u_i,v_i)\}\cup T[v_i,v_i']\cup SC_H[v_i',v_j']\cup T[v_j',v_j]$. This path is within $H_i:=\left(H\backslash\{e_2\}\right)\cup\{(u_i,v_i)\}$ but $\dist_{H_i}(u_j,v_j)>k$. Therefore, using $|E_{p'}|>k$ we get \Cref{combine}. From these, we have:
\begin{align*}
\maex(i,j)=|E_{T[u_i,b_1]}|+|E_{T[v_i',v_i]}|=\frac{(\ref{imore})+((\ref{combine})-(\ref{jless}))}{2}\ge \frac{k+1+2x-L_H}{2}\\
\maex(j,i)=|E_{T[u_j,b_1]}|+|E_{T[v_j',v_j]}|=\frac{(\ref{jmore})+((\ref{combine})-(\ref{iless}))}{2}\ge \frac{k+1+2x-L_H}{2}
\end{align*}
Since $|E_{SC_H[u_g',v_g']}|$ is minimized by definition, we have $|E_{SC_H[u_g',v_g']}|\leq x$, and the above calculation derives our statement.

\item[(3)] If $V_{T[v_i,v_i']}$ intersects with $V_{T[v_j,v_j']}$ and so does $V_{T[u_i,u_i']}$ with $V_{T[u_j,u_j']}$, By property of the tree we can conclude $u_i'=u_j',v_i'=v_j'$ and get the third graph of \Cref{casesover}. Similarly, we define $b_1,b_2$ as the lowest common ancestors of $u_i,u_j$ and $v_i,v_j$ respectively, and let $x:=|E_{SC_H[u_i',v_i']}|,y_u:=|E_{T[b_1,u_i']}|,y_v:=|E_{T[b_2,v_i']}|$. We can write the following inequalities
\begin{align}
\label{jless2}|E_{T[u_j,b_1]}|+|E_{T[v_j,b_2]}|&\leq k-x-y_u-y_v\\
\label{imore2}|E_{T[u_i,b_1]}|+|E_{T[v_i,b_2]}|&\ge k+1+x-L_H-y_u-y_v\\
\label{combine2}|E_{T[u_j,b_1]}|+|E_{T[b_1,u_i]}|+|E_{T[v_i,b_2]}|+|E_{T[b_2,v_j]}|&\ge k
\end{align}
These three inequalities \Cref{jless2}, \Cref{imore2} and \Cref{combine2} are derived in very similar ways as for \Cref{jless}, \Cref{imore} and \Cref{combine} respectively, so we omit details here. By a similar argument, we can show:
\begin{equation*}
\maex(i,j)=|E_{T[u_i,b_1]}|+|E_{T[b_2,v_i]}|=\frac{(\ref{imore2})+((\ref{combine2})-(\ref{jless2}))}{2}\ge \frac{k+1+2x-L_H}{2}
\end{equation*}
By a completely symmetric argument, we can get $\maex(j,i)\ge\frac{k+1+2x-L_H}{2}$. Since $|E_{SC_H[u_g',v_g']}|\leq x$. We can derive our statement.
\end{itemize}
\end{proof}
\begin{figure}[t!]
\centering 
\includegraphics[width=1.0\textwidth]{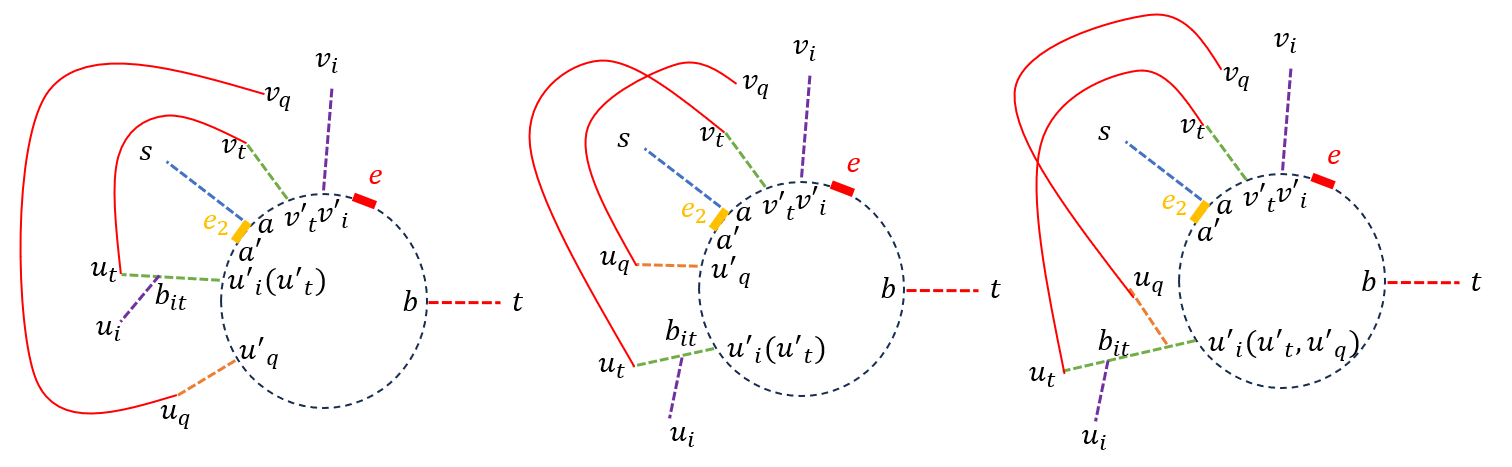} 
\caption{Three scenarios when $u_i'=u_t'$ for some $t\in\{g,h\}$. The red curves denote the original edges in $E_G$. Note that we don't identify and in fact don't care about the specific location of $T[v_q,v_q']$ so we don't draw it.} 
\label{cases2approx} 
\end{figure}
\begin{fact}\label{threepaths}
For any distinct $i,j,k\leq m$, we have
\begin{align*}
&|E_{T[u_i,u_i']}\backslash\left(E_{T[u_j,u_j']}\cup E_{T[u_k,u_k']}\right)|=\min\left(|E_{T[u_i,u_i']}\backslash E_{T[u_j,u_j']}|,|E_{T[u_i,u_i']}\backslash E_{T[u_k,u_k']}|\right)\\
&|E_{T[v_i,v_i']}\backslash\left(E_{T[v_j,v_j']}\cup E_{T[v_k,v_k']}\right)|=\min\left(|E_{T[v_i,v_i']}\backslash E_{T[v_j,v_j']}|,|E_{T[v_i,v_i']}\backslash E_{T[v_k,v_k']}|\right)
\end{align*}
\end{fact}
\begin{proof}
These two equations are symmetric, we only show the first one. Suppose if $T[u_i,u_i']$ doesn't intersect with at least one of $T[u_j,u_j']$ and $T[u_k,u_k']$. We are done. Otherwise, Since for any distinct $i,j\leq m$, if $T[u_i,u_i']$ shares some common edge with $T[u_j,u_j']$, there must be $u_i'=u_j'$. Therefore, $E_{T[u_i,u_i']}\backslash E_{T[u_j,u_j']}$ must form a tree path $T[u_i,b_{ij}]$ where $b_{ij}$ is the lowest common ancestor of $u_i$ and $u_j$. Therefore, we have
\begin{align*}
|E_{T[u_i,u_i']}\backslash\left(E_{T[u_j,u_j']}\cup E_{T[u_k,u_k']}\right)|&=|E_{T[u_i,b_{ij}]}\cap E_{T[u_i,b_{ik}]}|\\&=\min\{|E_{T[u_i,b_{ij}]}|,|E_{T[u_i,b_{ik}]}|\}\\
&=\min\left(|E_{T[u_i,u_i']}\backslash E_{T[u_j,u_j']}|,|E_{T[u_i,u_i']}\backslash E_{T[u_k,u_k']}|\right)
\end{align*}
\end{proof}
With the above two lemmas, we are ready to prove for any $i\in R$, $\dist_{H'}(u_i,v_i)\leq k$, which derives a contradiction and therefore completes the whole proof. Fix any $i\in R$, there are two cases:
\begin{itemize}
\item[(1)] If for some $t\in\{g,h\}$ and another $q:=\{g,h\}\backslash\{t\}$, there is $u_i'=u_t'$ and $b_{it}$ is not in $V_{T[u_{q},u_{q}')}$ (Symmetric for the case of $v_i'=v_t'$ and we omit here), where $b_{it}$ is the lowest common ancestor of $u_i$ and $u_t$. By definition of $t\in\{g,h\}$, $|E_{SC_H[u_t',v_t']}|$ is minimized, so $|E_{SC_H[u_i',v_i']}|=|E_{SC_H[u_t',v_i']}|\ge |E_{SC_H[u_t',v_t']}|$ and $v_i'\in V_{SC_H[v_t',b]}$. The three corresponding scenarios are shown in \Cref{cases2approx}.  Consider the path $p'=T[u_i,b_{it}]\cup T[b_{it},u_t]\cup\{(u_t,v_t)\}\cup T[v_t,v_t']\cup SC_H[v_t',v_i']\cup T[v_i',v_i]$ between $u_i$ and $v_i$. This path is in $H'$ and only uses edges in $T$ except for the one newly added edge $(u_t,v_t)$. Moreover, we can observe that
\begin{itemize}
\item[(a)] $p'$ doesn't touch edges in $E_{SC_H[u_t',v_t']}$ which are also tree-edges. By definition of $g$ we have $|E_{SC_H[u_t',v_t']}|\ge |E_{SC_H[u_g',v_g']}|$
\item[(b)] Since $v_i'\in V_{SC_H[v_t',b]}$, $p'$ doesn't touch cycle-edges in $E_{SC_H[v_i',u_t']}$, $|E_{SC_H[v_i',u_t']}|=L_H-|E_{SC_H[u_i',v_i']}|\ge L_H-|E_{SC_H[a,b]}|$. The last inequality is basically the same as \Cref{cyclepartup}. In these avoided edges at most one of them isn't tree-edge. Obviously, these edges are disjoint with the previous set of edges avoided.
\item[(c)] By a basically same argument as \Cref{innernon} and \Cref{doublenon}, $p'$ doesn't use any edges on $p[b,t]$, which are all tree-edges and not cycle-edges. By  \Cref{routelb}, we know $|E_{p[b,t]}|\ge \frac{k+|E_{SC_H[a,b]}|-L_H}{2}$. Since the previous two classes are all cycle-edges, this set is also disjoint with the previous sets of edges avoided.
\item[(d)] $p'$ doesn't touch edges in $E_{T[u_q,u_q']}$ and $E_{T[v_q,v_q']}\backslash\left(E_{T[v_i,v_i']}\cup E_{T[v_t,v_t']}\right)$. The number of edges not in $p'$ in this case is at least:
\begin{align}
&|E_{T[u_q,u_q']}|+|E_{T[v_q,v_q']}\backslash\left(E_{T[v_i,v_i']}\cup E_{T[v_t,v_t']}\right)|\\
\label{firsteq}\ge& |E_{T[u_q,u_q']}|+\min\left(|E_{T[v_q,v_q']}\backslash E_{T[v_i,v_i']}|,|E_{T[v_q,v_q']}\backslash E_{T[v_t,v_t']}|\right)\\
\ge&\min\left(\maex(q,i),\maex(q,t)\right)\\
\label{secondeq}\ge&\frac{k+1+|E_{SC_H[u_g',v_g']}|-L_H}{2}
\end{align}
\Cref{firsteq} is from \Cref{threepaths} and \Cref{secondeq} is from \Cref{routelbplus} since $q\in\{g,h\}$ and $i,t,q\in \{g,h\}\cup R$. These avoided edges are all tree-edges and not on the cycle so they are disjoint with the first and second set of edges avoided above. Moreover, by a similar argument to \Cref{innernon} and \Cref{doublenon}, this set is disjoint with the third set of avoided edges just discussed. 
\end{itemize}
All of the four classes of avoided edges discussed above are pairwise disjoint and are all tree-edges except for at most one, therefore. These edges are not in $p'$ and $p'$ is in the tree $T$. Therefore, the length of $p'$ can be bounded.
\begin{align*}
|E_{p'}|
\leq&1+|E_T|-|E_{SC_H[u_g',v_g']}|-(L_H-|E_{SC_H[a,b]}|-1)\\
&-\frac{k+|E_{SC_H[a,b]}|-L_H}{2}-\frac{k+1+|E_{SC_H[u_g',v_g']}|-L_H}{2}\\
\leq&n-\frac{3}{2}|E_{SC_H[u_g',v_g']}|-k+\frac{|E_{SC_H[a,b]}|}{2}+1\\
\leq&n-k+\frac{k}{4}+2
\end{align*}
The last inequality is from the property of shortest paths that $|E_{SC_H[a,b]}|\leq \frac{L_H}{2}\leq\frac{k+1}{2}$. When $k>\frac{4}{7}n+O(1)$, we have $|E_{p'}|\leq k$. Since $p'$ is in $H'$, we get $\dist_{H'}(u_i,v_i)\leq k$.

\item[(2)] If the above condition doesn't hold. Then either $u_i'\neq u_h'$ or $u_i'=u_h'$ and $b_{ih}\in V_{T[u_g,u_g']}$ (Symmetric for $v_i,v_h,v_g$ e.t.c). In both cases we have $T[u_h,u_h']\backslash T[u_g,u_g']$ doesn't edge-intersect with $T[u_i,u_i']$ (Symmetric for $v_i,v_h,v_g$ e.t.c). Consider the path $p'=T[u_i,u_i']\cup T[u_i',u_g']\cup T[u_g',u_g]\cup \{(u_g,v_g)\}\cup T[v_g,v_g']\cup T[v_g',v_i']\cup T[v_i',v_i]$. We can see $p'$ is a path on $H'$ and uses only one edge $(u_g,v_g)$ that is not on tree $T$. Let's consider which tree-edges it doesn't use.
\begin{itemize}
\item[(a)] By a basically same argument as \Cref{innernon} and \Cref{doublenon} as in the previous case, $p'$ doesn't use any edges on $p[b,t]$, which are all tree-edges and not cycle-edges. By \Cref{routelb}, we know $|E_{p[b,t]}|\ge \frac{k+|E_{SC_H[a,b]}|-L_H}{2}$
\item[(b)] $p'$ doesn't use edges in $E_{T[u_h,u_h']}\backslash\left(E_{T[u_g,u_g']}\cup E_{T[u_i,u_i']}\right)$ and $E_{T[v_h,v_h']}\backslash\left(E_{T[v_g,v_g']}\cup E_{T[v_i,v_i']}\right)$. They are all tree edges and not cycle-edges. Moreover, by a simple case analysis, in this case we have $E_{T[u_h,u_h']}\backslash\left(E_{T[u_g,u_g']}\cup E_{T[u_i,u_i']}\right)=E_{T[u_h,u_h']}\backslash E_{T[u_g,u_g']}$ and the similar result for $v_i,v_h,v_g$ e.t.c. Therefore, this set of edges is $\left(E_{T[u_h,u_h']}\cup E_{T[v_h,v_h']}\right)\backslash \left(E_{T[u_g,u_g']}\cup E_{T[v_g,v_g']}\right)$, which is exactly $\maex(h,g)$. From \Cref{routelbplus}, this set has at least $\frac{k+1+2|E_{SC_H[u_g',v_g']}|-L_H}{2}$ edges. By \Cref{innernon} and \Cref{doublenon}, it's disjoint with the previous set of avoided edges.
\item[(c)] In terms of cycle-edges, by the property of trees, $p'$ only touches cycle-edges in the edge-set $E_{T[u_i',a]}\cup E_{T[u_g',a]}\cup E_{T[v_i',a]}\cup E_{[v_g',a]}$. By definition of the tree $T$ and $u_i',v_i',u_g',v_g'$, we know $E_{T[u_i',a]}\cup E_{T[v_i',a]}=E_{SC_H[u_i',v_i']}$ and $E_{T[u_g',a]}\cup E_{[v_g',a]}=E_{SC_H[u_g',v_g']}$. It means $p'$ doesn't touch at least $L_H-|E_{SC_H[u_i',v_i']}|-|E_{SC_H[u_g',v_g']}|\ge L_H-|E_{SC_H[a,b]}|-|E_{SC_H[u_g',v_g']}|$ cycle-edges by \Cref{cyclepartup}, and at most one of cycle-edge avoided is not tree-edge. Since they are all cycle-edges, this set of edges is disjoint with the previous two avoided sets.
\end{itemize}
All three classes above of edges avoided in $p'$ are pairwise disjoint and are all tree-edges except for at most one. Therefore, the length of $p'$ can be bounded as
\begin{align*}
|E_{p'}|\leq&1+|E_T|-\frac{k+|E_{SC_H[a,b]}|-L_H}{2}-\frac{k+1+2|E_{SC_H[u_g',v_g']}|-L_H}{2}\\
&-(L_H-|E_{SC_H[a,b]}|-|E_{SC_H[u_g',v_g']}|)+1\\
\leq&n+3-k+\frac{|E_{SC_H[a,b]}|}{2}\\
\leq&n+3-k+\frac{1}{4}k
\end{align*}
The last inequality is from $|E_{SC_H[a,b]}|\leq \frac{L_H}{2}\leq\frac{k+1}{2}$. When $k>\frac{4}{7}n+O(1)$, we have $|E_{p'}|\leq k$. Since $p'$ is in $H'$, we get $\dist_{H'}(u_i,v_i)\leq k$.
\end{itemize}
This works for all $i\in R$, so we proved $H'$ is a $k$-spanner.
\end{proof}
Finally, we are ready to prove \Cref{2approxub}.

\begin{theorem}[Restated, \Cref{2approxub}]\label{2approxubit}
        There is a deterministic polynomial time algorithm $A$ such that for all sufficiently large $n$ and any $k>\frac{4}{7}n+O(1)$, given any $n$-vertex graph $G$ and its $k$-spanner $H=(V,E_H)\subseteq G$, $A(G,H,k)$ outputs a $k$-spanner $R=(V,E_R)$ of $G$ with girth at least $k+2$. Moreover, the size of 
 $R$ satisfies $|E_R|-n\leq2(E_H-n)+1$.  
\end{theorem}
\begin{proof}
Given any spanner $H$, we first find its smallest cycle and see whether $L_H>k+1$. If so, we are done. Otherwise, we apply \Cref{2approxmod} to break the smallest cycle while preserving the modified graph is still a $k$-spanner, with at most one more edge, and no newly added cycle with length at most $k+1$. By repeatedly doing the above procedure, we can finally find a subgraph $R\subseteq G$ with girth at least $k+2$ and is still a $k$-spanner of $G$. In each iteration, we remove an edge in a cycle with length at most $k+1$ and no cycle with length at most $k+1$ is newly generated. Therefore, all deleted edges are original edges from $E_H$ and every edge won't be deleted twice. Now let $H_i$ denote the graph modified from $H$ after $i$ iterations and let $H'_i=H_i\cap \left(H\backslash E_i\right)$ denote the subgraph consisting of edges untouched throughout the first $i$ iterations, where $E_i$ is the set of edges ever been added in the first $i$ iterations. Initially, $H'_0=H$. In the $i$-th iteration we delete an edge $e_2\in E_{SC_{H_{i-1}}}$ from $H_{i-1}$'s smallest cycle and moreover, as discussed above that smallest cycle $SC_{H_{i-1}}$ is totally within $H'_{i-1}$. Therefore, as long as $H'_{i-1}$ is connected, $H'_i$ is still connected. Moreover, we can see that $|E_{H'_i}|=|E_{H'_{i-1}}|-1$, so for any $i$, $H'_i$ is a connected subgraph such that $|E_{H'_{i}}|=|E_H|-i$. Therefore, it is impossible to get index $i>|E_H|-n+1$ since otherwise there must be $|E_{H'_{i-1}}|\leq n-1$ and $H'_{i-1}$ must be a tree by connectedness. This implies $H'_{i-1}$ doesn't contain any cycle and we terminate after $i-1$ iterations. Therefore, we will do at most $|E_H|-n+1$ iterations so the size blow-up is at most $|E_H|-n+1$. In conclusion, we have $|E_R|\leq |E_H|+(|E_H|-n+1)$
\end{proof}
\subsection{Upper Bound for Constant-Approx Good Pair with Smaller Stretch}\label{sec:moreapprox}
In this paragraph, we aim to prove \Cref{moreapproxub}. Our proof will generalize the idea in the proof of \Cref{2approxub} by classifying the $e_2$-endangered pairs $(u_1,v_1),\dots,(u_m,v_m)$ into different `buckets', and then add at most two $e_2$-endangered pairs from each of the `buckets'. Next, we will illustrate how we set up those buckets and what we do in each iteration.

\begin{lemma}\label{moreapproxmod}
Let $t\in\{1,2,3,4\}$. If $k>\frac{4t}{9t-4}n+O(1)$, and $H$ is a $k$-spanner of $G$ with smallest cycle of length at most $k+1$, we can find at most $r=2t^2$ $e_2$-endangered pairs $M:=\{(u_1,v_1),\dots,(u_r,v_r)\}\in E_G$ such that $H'=\left(H\backslash\{e_2\}\right)\cup M$ is still a $k$-spanner, and the newly added edges in $H'$ are not in any cycle with length at most $k+1$.
\end{lemma}
\begin{proof}
Since some technical parts of this proof are similar to proof of \Cref{2approxmod}, we will sketch the overlap part, and focus more on new ideas from this proof.

Let $R=\{(u_1,v_1),\dots,(u_m,v_m)\}$ be the set of $e_2$-endangered pairs in $H$. We will partition it into $r=t^2$ disjoint subsets $\{R_{i,j}\}_{i,j\in[t]}$. They are defined as
\[R_{i,j}:=\left\{(u_k,v_k):k\in[m],|E_{T[u_k',a]}|\in\big((i-1)B,iB\big],|E_{T[v_k',a']}|\in\big((j-1)B,jB\big]\right\},\text{ where }B:=\lceil\frac{|E_{SC_H[a,b]}|}{t}\rceil\]
By a similar argument as \Cref{cyclepartup}, it can be shown that $\{R_{i,j}\}_{i,j\in[t]}$ form a partition of $R$. Let $G_{i,j}:=\left(G\backslash R\right)\cup R_{i,j}$. We will solve each $G_{i,j}$ using the same procedure as in \Cref{2approxmod} separately, and then combine them since $G=\bigcup_{i,j\in[t]}G_{i,j}$. However, At this time we can take advantage of the special property of $G_{i,j}$, and get a better bound for each subtask on $G_{i,j}$.
\begin{lemma}\label{divide}
For any $i,j\in[t]$, if $k>\frac{4t}{9t-4}n+O(1)$, and $H$ is a $k$-spanner of $G_{i,j}$ with smallest cycle of length at most $k+1$, we can find at most two $e_2$-endangered pairs $(u^{(i,j)}_g,v^{(i,j)}_g),(u^{(i,j)}_h,v^{(i,j)}_h)\in R_{i,j}$ such that $H'_{i,j}=\left(H\backslash\{e_2\}\right)\cup\{(u^{(i,j)}_g,v^{(i,j)}_g),(u^{(i,j)}_h,v^{(i,j)}_h)\}$ is still a $k$-spanner of $G_{i,j}$, and the newly added edges in $H'_{i,j}$ are not in any cycle with length at most $k+1$.
\end{lemma}
\begin{proof}
We use exactly the same construction as in \Cref{2approxmod} on $H$ and $G_{i,j}$. It suffices to build a better bound $k>\frac{4t}{9t-4}n+O(1)$ by extracting the special property of $G_{i,j}$. All notations in this proof will follow from the proof of \Cref{2approxmod}. For any $(u^{(i,j)}_k,v^{(i,j)}_k)\in R_{i,j}$, we aim to prove $\dist_{H'_{i,j}}(u^{(i,j)}_k,v^{(i,j)}_k)\leq k$, which completes our proof.

Note that there are two cases in the proof of \Cref{2approxmod} and we analyzed them separately. There are the same two cases here
\begin{itemize}
\item[(1)] If for some $w\in\{g,h\}$ and another $q\in\{g,h\}\backslash\{w\}$, there is $u'^{(i,j)}_k=u'^{(i,j)}_w$ and $b_{kw}$ is not in $V_{T[u^{(i,j)}_q,u'^{(i,j)}_q)}$, where $b_{kw}$ is the lowest common ancestor of $u^{(i,j)}_k$ and $u^{(i,j)}_w$. We can get similar arguments about $p'$-avoided edges in $T$. The only difference is now the size of the 
 cycle-part in $E_{p'}\cap E_{SC_H}=E_{T[v'^{(i,j)}_w,v'^{(i,j)}_k]}$ can be bounded by \[|E_{T[v'^{(i,j)}_w,v'^{(i,j)}_k]}|\leq B+1\leq \frac{|E_{SC_H[a,b]}|}{t}+2\]
Therefore, we can improve the bound of $|E_{p'}|$ as
\begin{align*}
|E_{p'}|
\leq&1+|E_T|-(L_H-\frac{|E_{SC_H[a,b]}|}{t}-2)\\
&-\frac{k+|E_{SC_H[a,b]}|-L_H}{2}-\frac{k+1+|E_{SC_H[u'^{(i,j)}_g,v'^{(i,j)}_g]}|-L_H}{2}\\
\leq&n-\frac{1}{2}|E_{SC_H[u'^{(i,j)}_g,v'^{(i,j)}_g]}|-k+(\frac{1}{t}-\frac{1}{2})|E_{SC_H[a,b]}|+O(1)\\
\leq&n-k+\frac{k}{2}\max(0,\frac{1}{t}-\frac{1}{2})+O(1)
\end{align*}
Since $t\in\{1,2,3,4\}$, we can verify that $|E_{p'}|\leq k$ holds when $k>\frac{4t}{9t-4}n+O(1)$, so the first case is finished.
\item[(2)] If it isn't in the first case, then either $u'^{(i,j)}_k\neq u'^{(i,j)}_h$ or $u'^{(i,j)}_k=u'^{(i,j)}_h$ and $b_{kh}\in V_{T[u^{(i,j)}_g,u'^{(i,j)}_g]}$. We can get similar arguments about $p'$-avoided edges in $T$. The only difference is now the size of the cycle-part in $E_{p'}\cap E_{SC_H}=E_{T[v'^{(i,j)}_g,v'^{(i,j)}_k]}\cup E_{T[u'^{(i,j)}_g,u'^{(i,j)}_k]}$ can be bounded by \[|E_{T[v'^{(i,j)}_w,v'^{(i,j)}_k]}|+|E_{T[u'^{(i,j)}_g,u'^{(i,j)}_k]}|\leq 2(B+1)\leq \frac{2|E_{SC_H[a,b]}|}{t}+4\]
Therefore, we can improve the bound of $|E_{p'}|$ as
\begin{align*}
|E_{p'}|\leq&1+|E_T|-\frac{k+|E_{SC_H[a,b]}|-L_H}{2}-\frac{k+1+2|E_{SC_H[u'^{(i,j)}_g,v'^{(i,j)}_g]}|-L_H}{2}\\
&-(L_H-\frac{2|E_{SC_H[a,b]}|}{t}-4)+1\\
\leq&n-k+(\frac{2}{t}-\frac{1}{2})|E_{SC_H[a,b]}|+O(1)\\
\leq&n-k+(\frac{1}{t}-\frac{1}{4})k+O(1)
\end{align*}
Since $t\in\{1,2,3,4\}$, we can verify that $|E_{p'}|\leq k$ holds when $k>\frac{4t}{9t-4}n+O(1)$, we finish the proof for both cases.
\end{itemize}
\end{proof}
Let $H'=\bigcup_{i,j\in[t]}H'_{i,j}$, where $H'_{i,j}$ is the transformed $k$-spanner of $G_{i,j}$ by invoking \Cref{divide} for each $i,j\in[t]$. Since $H'_{i,j}$ is a $k$-spanner of $G_{i,j}$ for each $i,j\in[t]$, we know $H'$ is also a $k$-spanner of $G$. Note that $H'$ may create new small cycles with length at most $k+1$. However, we can slightly modify our construction to avoid it. Instead of partitioning $R$ into $\{R_{i,j}\}_{i,j\in[t]}$ and running \Cref{divide} on each $(i,j)$ concurrently, we should generate $R_{i,j}$ and  $H'_{i,j}$ one by one. Specifically, there are at most $r=t^2$ pairs $(i_1,j_1),\dots,(i_r,j_r)\in[t]\times[t]$ in total. For any $d\leq r$, let $H'_d$ denote the transformed $k$-spanner after processing the first $d$ pairs. Suppose $H'_d$ doesn't create new small cycles with length at most $k+1$. Then when processing $(i_{d+1},j_{d+1})$, we consider $R$ and $R_{i_{d+1},j_{d+1}}$ as the corresponding set of $e_2$-endangered pairs on $H'_d\cup \{e_2\}$, not on the original spanner $H$. Then, we use $R_{i_{d+1},j_{d+1}}$ generated in this way to find $(u^{(i_{d+1},j_{d+1})}_g,v^{(i_{d+1},j_{d+1})}_g),(u^{(i_{d+1},j_{d+1})}_h,v^{(i_{d+1},j_{d+1})}_h)\in R_{i_{d+1},j_{d+1}}$ and invoke \Cref{divide} to generate $H'_{i_{d+1},j_{d+1}}$. Let $H'_{d+1}:=H'_d\cup H'_{i_{d+1},j_{d+1}}$, we finish the procedure on the $(d+1)$-th pair. In this way, by proceeding the $t^2$ pairs sequentially, we can show that the final transformed $k$-spanner $H'=H'_r$ doesn't create any new small cycles with length at most $k+1$. 
\end{proof}
Similar to \Cref{2approxub}, we can finally derive \Cref{moreapproxub} using the procedure \Cref{moreapproxmod}.
\begin{theorem}[Restated, \Cref{moreapproxub}]\label{moreapproxubit}
    For $t\in\{1,2,3,4\}$, there is a deterministic polynomial time algorithm $A$ such that for all sufficiently large $n$ and any $k>\frac{4t}{9t-4}n+O(1)$, given any $n$-vertex graph $G$ and its $k$-spanner $H=(V,E_H)\subseteq G$, $A(G,H,k)$ outputs a $k$-spanner $R=(V,E_R)$ of $G$ with girth at least $k+2$. Moreover, the size of 
 $R$ satisfies $|E_R|-n\leq2t^2(E_H-n)+2t^2$.  
\end{theorem}
\begin{proof}
Using \Cref{moreapproxmod}, this theorem follows from exactly the same proof as in \Cref{2approxubit}. We omit it here for brevity.
\end{proof}
\begin{remark}\label{cannotbeyond}
We should explain the reason why we cannot set $t\ge 5$ and get bounds when $k\leq \frac{1}{2}n$ in \Cref{moreapproxubit}. In fact, even if we set the bucket parameter $t\ge 5$, the calculations in \Cref{divide} won't derive better bound than $|E_{p'}|\leq n-k$, which satisfies $|E_{p'}|\leq k$ only when $k>\frac{1}{2}n$.
\end{remark}
\section{Future Direction}
\begin{itemize}
\item[(1)] We still don't know any (approximately) universal optimality argument when $\frac{1}{3}n<k<\frac{1}{2}n$. It is interesting to explore how well the greedy algorithm works in this range.
\item[(2)] It is interesting to generalize our results to weighted graphs. Currently, it's not quite clear how to do that.
\end{itemize}

\section*{Acknowledgments}
We would like to thank Greg Bodwin for his wonderful course EECS598 ``Theory of Network Design'' that led the author to this area. We also think Greg for helpful discussions including his advice on exploring ``approximately universal optimality''. Besides, we thank Benyu Wang and Gary Hoppenworth for insightful discussions. In particular, Benyu Wang gave us some inspirations in proving \cref{largecycle}.
{\small \bibliography{main}}\newpage

\newcommand{\etalchar}[1]{$^{#1}$}
\begin{thebibliography}{DJWW22}

\bibitem[ADD{\etalchar{+}}93]{add93}
Ingo Alth{\"{o}}fer, Gautam Das, David~P. Dobkin, Deborah Joseph, and Jos{\'{e}} Soares.
\newblock On sparse spanners of weighted graphs.
\newblock {\em Discret. Comput. Geom.}, 9:81--100, 1993.
\newblock \href {https://doi.org/10.1007/BF02189308} {\path{doi:10.1007/BF02189308}}.

\bibitem[AHL02]{moore}
Noga Alon, Shlomo Hoory, and Nathan Linial.
\newblock The moore bound for irregular graphs.
\newblock {\em Graphs Comb.}, 18(1):53--57, 2002.
\newblock URL: \url{https://doi.org/10.1007/s003730200002}, \href {https://doi.org/10.1007/S003730200002} {\path{doi:10.1007/S003730200002}}.

\bibitem[Awe85]{dis1}
Baruch Awerbuch.
\newblock Complexity of network synchronization.
\newblock {\em J. {ACM}}, 32(4):804--823, 1985.
\newblock \href {https://doi.org/10.1145/4221.4227} {\path{doi:10.1145/4221.4227}}.

\bibitem[BCLR86]{dis2}
Sandeep~N. Bhatt, Fan R.~K. Chung, Frank~Thomson Leighton, and Arnold~L. Rosenberg.
\newblock Optimal simulations of tree machines (preliminary version).
\newblock In {\em 27th Annual Symposium on Foundations of Computer Science, Toronto, Canada, 27-29 October 1986}, pages 274--282. {IEEE} Computer Society, 1986.
\newblock \href {https://doi.org/10.1109/SFCS.1986.38} {\path{doi:10.1109/SFCS.1986.38}}.

\bibitem[BEG{\etalchar{+}}22]{ultrasparse2}
Marcel Bezdrighin, Michael Elkin, Mohsen Ghaffari, Christoph Grunau, Bernhard Haeupler, Saeed Ilchi, and V{\'{a}}clav Rozhon.
\newblock Deterministic distributed sparse and ultra-sparse spanners and connectivity certificates.
\newblock In Kunal Agrawal and I{-}Ting~Angelina Lee, editors, {\em {SPAA} '22: 34th {ACM} Symposium on Parallelism in Algorithms and Architectures, Philadelphia, PA, USA, July 11 - 14, 2022}, pages 1--10. {ACM}, 2022.
\newblock \href {https://doi.org/10.1145/3490148.3538565} {\path{doi:10.1145/3490148.3538565}}.

\bibitem[Bod24]{light4}
Greg Bodwin.
\newblock An alternate proof of near-optimal light spanners.
\newblock In Merav Parter and Seth Pettie, editors, {\em 2024 Symposium on Simplicity in Algorithms, {SOSA} 2024, Alexandria, VA, USA, January 8-10, 2024}, pages 39--55. {SIAM}, 2024.
\newblock \href {https://doi.org/10.1137/1.9781611977936.5} {\path{doi:10.1137/1.9781611977936.5}}.

\bibitem[CDNS95]{light2}
Barun Chandra, Gautam Das, Giri Narasimhan, and Jos{\'{e}} Soares.
\newblock New sparseness results on graph spanners.
\newblock {\em Int. J. Comput. Geom. Appl.}, 5:125--144, 1995.
\newblock \href {https://doi.org/10.1142/S0218195995000088} {\path{doi:10.1142/S0218195995000088}}.

\bibitem[CKL{\etalchar{+}}22]{graphalg1}
Li~Chen, Rasmus Kyng, Yang~P. Liu, Richard Peng, Maximilian~Probst Gutenberg, and Sushant Sachdeva.
\newblock Maximum flow and minimum-cost flow in almost-linear time.
\newblock In {\em 63rd {IEEE} Annual Symposium on Foundations of Computer Science, {FOCS} 2022, Denver, CO, USA, October 31 - November 3, 2022}, pages 612--623. {IEEE}, 2022.
\newblock \href {https://doi.org/10.1109/FOCS54457.2022.00064} {\path{doi:10.1109/FOCS54457.2022.00064}}.

\bibitem[CKR{\etalchar{+}}91]{chip1}
J.~Cong, A.~Kahng, G.~Robins, M.~Sarrafzadeh, and C.K. Wong.
\newblock Performance-driven global routing for cell based ics.
\newblock In {\em [1991 Proceedings] IEEE International Conference on Computer Design: VLSI in Computers and Processors}, pages 170--173, 1991.
\newblock \href {https://doi.org/10.1109/ICCD.1991.139874} {\path{doi:10.1109/ICCD.1991.139874}}.

\bibitem[CKR{\etalchar{+}}92]{chip2}
J.~Cong, A.B. Kahng, G.~Robins, M.~Sarrafzadeh, and C.K. Wong.
\newblock Provably good performance-driven global routing.
\newblock {\em IEEE Transactions on Computer-Aided Design of Integrated Circuits and Systems}, 11(6):739--752, 1992.
\newblock \href {https://doi.org/10.1109/43.137519} {\path{doi:10.1109/43.137519}}.

\bibitem[DJWW22]{graphalg2}
Mina Dalirrooyfard, Ce~Jin, Virginia~Vassilevska Williams, and Nicole Wein.
\newblock Approximation algorithms and hardness for n-pairs shortest paths and all-nodes shortest cycles.
\newblock In {\em 63rd {IEEE} Annual Symposium on Foundations of Computer Science, {FOCS} 2022, Denver, CO, USA, October 31 - November 3, 2022}, pages 290--300. {IEEE}, 2022.
\newblock \href {https://doi.org/10.1109/FOCS54457.2022.00034} {\path{doi:10.1109/FOCS54457.2022.00034}}.

\bibitem[ENS14]{light3}
Michael Elkin, Ofer Neiman, and Shay Solomon.
\newblock Light spanners.
\newblock In Javier Esparza, Pierre Fraigniaud, Thore Husfeldt, and Elias Koutsoupias, editors, {\em Automata, Languages, and Programming - 41st International Colloquium, {ICALP} 2014, Copenhagen, Denmark, July 8-11, 2014, Proceedings, Part {I}}, volume 8572 of {\em Lecture Notes in Computer Science}, pages 442--452. Springer, 2014.
\newblock \href {https://doi.org/10.1007/978-3-662-43948-7\_37} {\path{doi:10.1007/978-3-662-43948-7\_37}}.

\bibitem[FS20]{existopt}
Arnold Filtser and Shay Solomon.
\newblock The greedy spanner is existentially optimal.
\newblock {\em {SIAM} J. Comput.}, 49(2):429--447, 2020.
\newblock \href {https://doi.org/10.1137/18M1210678} {\path{doi:10.1137/18M1210678}}.

\bibitem[Pet08]{ultrasparse}
Seth Pettie.
\newblock Distributed algorithms for ultrasparse spanners and linear size skeletons.
\newblock In Rida~A. Bazzi and Boaz Patt{-}Shamir, editors, {\em Proceedings of the Twenty-Seventh Annual {ACM} Symposium on Principles of Distributed Computing, {PODC} 2008, Toronto, Canada, August 18-21, 2008}, pages 253--262. {ACM}, 2008.
\newblock \href {https://doi.org/10.1145/1400751.1400786} {\path{doi:10.1145/1400751.1400786}}.

\bibitem[PS89]{spanner_intro}
David Peleg and Alejandro~A. Sch{\"{a}}ffer.
\newblock Graph spanners.
\newblock {\em J. Graph Theory}, 13(1):99--116, 1989.
\newblock URL: \url{https://doi.org/10.1002/jgt.3190130114}, \href {https://doi.org/10.1002/JGT.3190130114} {\path{doi:10.1002/JGT.3190130114}}.

\end{thebibliography}

\listoffixmes

\end{document}